\documentclass[ukenglish,envcountsame,envcountresetchap,graybox]{svmult}
\usepackage{type1cm}       
\usepackage{graphicx}        
\usepackage{multicol}        
\usepackage[bottom]{footmisc}

\usepackage{newtxtext}       %

\makeindex             

\usepackage{url}    
\usepackage[british]{babel}
\usepackage{amssymb,amsmath,braket,array,comment}
\usepackage{samskymacros}
\usepackage{microtype}

\usepackage{soul}

\usepackage{tikz}
\usepackage{tikz-cd}
\usepackage{tikzsymbols}
\usetikzlibrary{calc}
\usepackage{enumerate}
\usepackage{hyperref}
\usepackage{xcolor}
\usepackage{graphicx}
\usepackage[margin=2.5cm]{geometry}
\usepackage{xfrac}
\usepackage{cite}

\usetikzlibrary{matrix}

\begin{document}

\title*{Closing Bell}
\subtitle{Boxing black box simulations in the  resource theory of contextuality}
\titlerunning{Closing Bell: Boxing black box simulations in the resource theory of contextuality}
\toctitle{Closing Bell: Boxing black box simulations in the resource theory of contextuality}
\author{Rui Soares Barbosa, Martti Karvonen, and Shane Mansfield}
\authorrunning{R.S.~Barbosa \and M.~Karvonen \and S.~Mansfield} 
\institute{
Rui Soares Barbosa
\at INL -- International Iberian Nanotechnology Laboratory, Braga, Portugal \\
\email{rui.soaresbarbosa@inl.int}
\and
Martti Karvonen
\at Department of Mathematics and Statistics, University of Ottawa, Ottawa, Canada\\
\email{martti.karvonen@uottawa.ca}
\and
Shane Mansfield
\at
Quandela, Palaiseau, France \\
\email{shane.mansfield@quandela.com}
}
\maketitle

\abstract{
This chapter contains an exposition of the sheaf-theoretic framework for contextuality emphasising resource-theoretic aspects, as well as some original results on this topic.
In particular, we consider functions that transform empirical models on a scenario $S$ to empirical models on another scenario $T$, and characterise those that are induced by classical procedures between $S$ and $T$ corresponding to `free' operations in the (non-adaptive) resource theory of contextuality.
We construct a new `hom' scenario built from $S$ and $T$, whose empirical models induce such functions.
Our characterisation then boils down to being induced by a non-contextual model.
We also show that this construction on scenarios provides a closed structure on the category of measurement scenarios.
}

\section*{Prologue}

Among the many and varied facets of Samson Abramsky's work have been his
contributions to the foundations of quantum mechanics. Approaching the subject
through the lenses of computer science, he has brought its modes of thought and
mathematical tools to bear on the analysis of natural systems, providing fresh
perspectives that have illuminated the fundamental structures of the quantum world
and their interaction with notions of information and computation.

In particular, some of Samson's major contributions over the past decade or so have
been to the study of non-locality and contextuality. These are key phenomena that set
quantum theory apart from classical physical theories and that can be shown to
relate closely to quantum-over-classical advantage in computation and
information processing. Samson, his collaborators, and others have developed a general,
unifying framework that shines light on structural and logical aspects at the
core of these phenomena
\cite{ab,cohomology-of-contextuality,abramsky2013relational,abramskyhardy:logical,abramsky2013robust,abramsky2014operational,abramsky2014classification,abramsky2014no,contextualitycohomologyparadox,abramsky2015contextualsemantics,abramsky2016hardy,abramsky2016possibilities,abramsky2017quantum,abramsky2018minimum,abramsky2017contextual,abramsky2017complete,ourlics:comonadicview,abramsky2019simulations,abramsky2019non,Samson:Pitowsky,abramsky2021logic,
wangetal:ambiguous,abramsky2018:borders,abramsky2022kconsistency,
barbosa2014monogamy,gogiosozeng:avn,barbosaetal:cv,boothetal:negativity,karvonen2018categories,karvonen2021neither,aasnaess2020,caru2017,caru2018,kishida2014,kishida2016,ruishizzle:extendability,mansfield2017unified,mansfield2018quantum,mansfield2017Hardy,mansfieldfritz2012,raussendorf2013contextuality,desilva2018logical}
Perhaps surprisingly, the crystallised formulation that emerges requires
remarkably little from the formalism of quantum theory. Stripped down to its
essentials, contextuality arises in the tension between local consistency and
global inconsistency which is made possible by the limitation that one
only has access to partial views of a quantum system. As it turns out, and as
Samson has duly pointed out elsewhere, similar structures and concerns crop up
in many other subjects, from relational databases and constraint satisfaction
to natural language \cite{abramsky2015contextualsemantics}.

Our contribution to this volume has two main objectives. On the one hand, it
presents some new results that are suggestive of how the research programme
might evolve over the coming years (Sections~\ref{sec:main-result} and \ref{sec:closure}). In addition to the results themselves we
provide detailed discussion of some
of the open questions and research directions that arise from them (Section~\ref{sec:outlook}).
On the other hand, it is partly aimed at providing an
up-to-date exposition of some of the main ideas of the framework (Sections~\ref{sec:framework-objects} and \ref{sec:framework-morphisms}). The adopted
perspective is -- we hope -- somewhat original. It differs from previous
expositions in that it focuses on the resource
theory of contextuality developed in our recent work with Samson.

This resource-theoretic perspective has placed the emphasis on simulations, or transformations, between
contextual behaviours rather than on individual instances of such behaviours;
\ie on morphisms rather than objects, to adopt the language of category
theory. We will show that a novel upshot of this perspective is a uniform
treatment of some important concepts from the literature: non-local games, for
example, arise as particular instances of simulations.

The new contributions
concern precisely the structure of this resource theory of contextuality. The
`free' transformations (\ie the classical simulations) between contextual
behaviours are characterised by regarding transformations as empirical
behaviours in their own right and reducing the question of `free'-ness to that
of non-contextuality of the corresponding behaviour.
In categorical language,
this is achieved through internalising the hom-sets, finding a closed structure in
the category of simulations. 

The technical contents of the chapter are preceded by a lengthy introductory
section, which aims to give a broad overview that motivates the approach and the basic ingredients of the
framework. The impatient and technically-minded reader may safely skip it and
jump straight into the weedier pastures of definition--theorem--proof land. The
more leisurely reader at the opposite extreme may be tempted to stop at the
introduction, in which we have aimed to convey the central ideas, and our
hope is that they won't leave empty-handed.

We have endeavoured to make the chapter accessible to a broad readership including
logicians, computer scientists, physicists of a foundational bent, and
superpositions thereof. In particular, no knowledge whatsoever of quantum
mechanics is assumed or even used: it really is all just about partial
information. We hope that everyone -- including Samson -- will be able to find something
of interest here.

\section{Introduction}\label{sec:introduction}

\subsection*{A behavioural lens}

Systems will be considered from a
purely \stress{observational} or \stress{operational} perspective.\footnotemark\
As such, a
system will simply be treated as a black box with which an external agent can
interact. Computer scientists may think of these interactions as the posing of
queries and the obtaining of responses, like in the calling of functions in a
programming language.
Physicists may prefer to think of interactions as
the performing of measurements and obtaining of outcomes. From this perspective
the states of a system are not defined \stress{intrinsically} or \textit{a
priori}, but rather they are descriptions of the empirically observable
behaviour in every allowed interaction.

\footnotetext{At the risk of provoking a relapse into an erstwhile indulgence of Samson's,
who has been known to self-identify as a recovering philosopher,
we remark that this chapter adopts an approach that is somewhat in
the empiricist tradition of philosophy. As someone whose research interests and contributions are
ever-evolving and indeed ever-relevant, one former student of his has also pointed
out that Samson has clearly distinguished himself from another famous Samson, the
dinosaur.}

In this, computer scientists may be familiar with the terms observational or
behavioural as opposed to \stress{state-space based} approaches, while physicists
may recognise a similar distinction between operationalism and realism. At the
same time, it may be worth commenting that adopting this perspective need not
entail any deep philosophical commitment, but merely a methodological one. For
on the one hand `operationalism is, at least, a useful exercise for freeing
the mind from the baggage of preconceptions about the world'
\cite{hardy2010physics}, while on the other it is an apt approach when one's
primary concern is the practical one of understanding how best to utilise the
systems under consideration, as is typical in quantum information and
computation.

Attention will be limited to \stress{single-use} black boxes. The
agent is permitted only one round of interaction with the system, which is
subsequently unusable. This kind of `one-shot' limitation is typical in
interactions with quantum systems, where measurements are often destructive,
rendering the system unavailable for future use. In terminology more familiar
to computer scientists, there is no state update or continuation. Equivalently,
we may think of black boxes that are reinitialised after each round of
interaction. Despite the apparent poverty of a setting shorn of sequential
interactions and transformations, it is in fact already rich enough to capture
salient features that set quantum systems apart from classical ones (indeed to
capture the most well-loved and well-studied of these features).

The interface, or \stress{type}, of a black box specifies a finite set of basic
queries or measurements and their respective sets of possible responses or outcomes.\footnote{From now on, we shall primarily adopt the terminology measurements and outcomes.}
The agent consumes the box by simultaneously performing a subset of these
measurements. Of specific interest will be situations in which not all subsets
of measurements can be performed jointly. Such a limitation could be imposed as
a design feature: \eg security concerns may dictate that some combinations of
database attributes not be simultaneously accessible. Similarly, it could
result from a lack of experimental resources: \eg having a limited number of
measurement devices or detectors in a physics experiment. However, our primary
motivation comes from quantum theory, where the limitation reflects a more
fundamental restriction stemming from the incompatibility of certain
combinations of measurements, in the sense of the uncertainty principle.
Performing one measurement may spoil the possibility of performing another by
disturbing its outcome statistics.

In such a situation it is not possible for the agent to freely examine \stress{all} the
observable properties of the system, as would be matter of course in the
setting of classical physics. Instead, their access to the system is limited,
mediated by sets of jointly-measurable observables, called \emph{contexts}.
Each context provides a partial (classical) window into the (quantum) system. And it is
only through the varied collection of partial points of view afforded by these
windows that one may infer about the system as a whole.

This limitation calls to mind the situation described in the opening passage of Peter Johnstone's compendium
on topos theory, \textit{Sketches of an Elephant} \cite{johnstone2002sketches}:
\begin{quotation}
  Four men, who had been blind from birth, wanted to know what
  an elephant was like; so they asked an elephant-driver for information. He
  led them to an elephant, and invited them to examine it; so one man felt the
  elephant's leg, another its trunk, another its tail and the fourth its ear.
  Then they attempted to describe the elephant to one another. The first man
  said `The elephant is like a tree'. `No', said the second, `the elephant is
  like a snake'. `Nonsense!' said the third, `the elephant is like a broom'.
  `You are all wrong,' said the fourth, `the elephant is like a fan'. And so
  they went on arguing amongst themselves, while the elephant stood watching
  them quietly. [\dots] But the important thing about the elephant is that
  `however you approach it, it is still the same animal'\dots
\end{quotation}
We might imagine that the debate would soon resolve itself if only each of the
participants were to move around the elephant and bit-by-bit -- or sketch-by-sketch --
build up a more
global picture of the animal.
Yet, the limitation that only partial empirical
information about the system can ever be obtained at once gives rise to an
altogether more intriguing set of possibilities. After all,
the elephant in this story might be thought of as an essentially
classical beast. And so in some respects this is the point at which the real
fun begins.

\subsection*{Contextuality: `at the borders of paradox'}

The idea is that even though the viewpoints afforded by overlapping
contexts may fit nicely together there can nevertheless be situations in which
it is impossible to paste \stress{all} of them into a consistent global
picture. Such a gap between \emph{local consistency} and \emph{global inconsistency}
gives rise to an apparent paradox,
an instance of which is beautifully illustrated
by M.C.\ Escher's lithograph \textit{Klimmen en dalen}
reproduced here in Figure~\ref{fig:escher}.\footnotemark

\footnotetext{This is similarly illustrated by the `impossible biscuit' in the
poster for the 2018 Lorentz Centre workshop \textit{Logical Aspects of Quantum
Information}, which was co-organised by Samson:
\url{https://www.lorentzcenter.nl/logical-aspects-of-quantum-information.html}}

\begin{figure}[tbhp]
  \centering
  \includegraphics[width=0.7\textwidth]{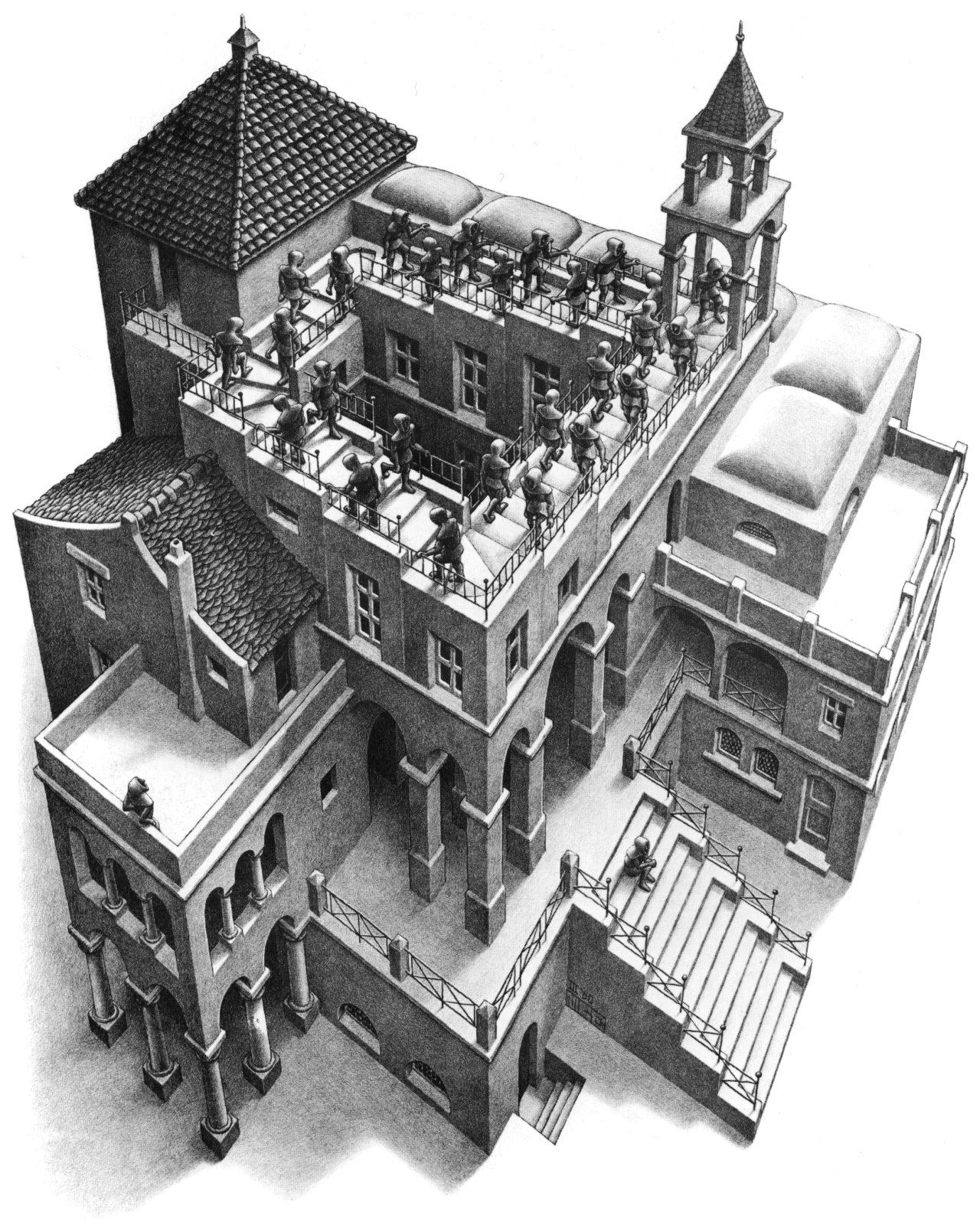}
  \caption{M. C. Escher, \textit{Klimmen en dalen} (Ascending and descending),
  1960. Lithograph, 285mm $\times$ 355mm.}
  \label{fig:escher}
\end{figure}

More concretely, the behaviour of a black box is described by an
\emph{empirical model} $e$ which details its (probabilistic)
response to any allowed interaction. For each allowed context of
measurements $C$, it specifies a probability distribution $e_C$ over the
set of available joint outcomes. We will encounter several examples of empirical
models in what follows (see Examples
\ref{ex:trianglemodel}--\ref{ex:possibilisticns}).

The distributions that make
up an empirical model are naturally constrained to be `locally' compatible in
the following sense. If $U$ is a subset of a context $C$, one could jointly
perform the measurements in $C$ but then forget the outcomes of those
measurements not in $U$. This would yield joint outcomes to the measurements in
$U$ with probabilities given by the marginal distribution $e_C|_U$.
The local compatibility requirement is that for any two contexts $C$ and $C'$
containing $U$, the marginals $e_C|_U$ and $e_{C'}|_U$ coincide. In other
words, $e_C$ and $e_{C'}$ \textit{agree on their overlap}, $e_C|_{C \cap C'} =
e_{C'}|_{C \cap C'}$. Thus, the probabilistic empirical behaviour observed for
the measurements in $U$ is the same regardless of whether they arise as a
subset of the context $C$ or of the context $C'$.

It is like saying that sketches of the
system should agree, or fit together, wherever they overlap. This is a property
that holds of empirical models that arise in physics. It both ensures a
compatibility with basic tenets of relativity in certain scenarios (see
discussion of Example~\ref{ex:ns}) and justifies the independent labelling of
individual measurements.

A more `\stress{global}' notion of compatibility asks that these locally
compatible probability distributions can be `glued' together into a global
probability distribution over joint outcomes to \stress{all} of the measurements at
once. Concretely, this global distribution would yield the empirical probability distribution $e_C$
when marginalised to each context $C$.

Existence of such a global distribution
allows one to think of the state of the system as a probabilistic mixture of
deterministic states that assign definite values to all observables.
Such deterministic states could be empirically inaccessible, which is why they are
often referred to as \stress{hidden} variables or as \stress{ontic} states (as
opposed to epistemic or empirical ones). But even if the agent is not in fact
allowed to perform all the measurements at once, they would have no
reason to doubt that these have
predetermined would-be outcome values at a more fundamental, though inaccessible,
level, and that these values exist independently of the agent's choice of
measurement context.
This is like asking that taken together the sketches provide
a coherent picture of the system as a whole -- or, to borrow Einstein's idiom,
asserting that the moon is there even when one is not looking at it
\cite{Pais1979:EinsteinQuantum,Mermin1985:moon}.
Yet, examples of violation of
the global consistency condition can be found. In the case of quantum systems,
they preclude such a conceptually neat and intuitive understanding of the
underlying states of the system, and thus of the physical world we inhabit.

Empirical models assign probabilities to observable events. When these
empirical probabilities satisfy local but not global consistency they are said
to be \emph{contextual}. As has been pointed out by Pitowsky
\cite{Pitowsky:Boole}, and further illuminated by Samson
\cite{Samson:Pitowsky}, one could say that contextuality was anticipated as far
back as in Boole's work on the `conditions of possible experience'
\cite{Boole1862}. Boole derived inequalities that the probabilities of
logically related events must satisfy in order to be, in our terms, globally
consistent. Of course, as the terminology indicates, Boole, sitting in his
rain-pelted quarters in what was then Queen's College, Cork, would have believed that
these conditions must be satisfied by real-world experiments. However, we now
know that this is not the case.

It was roughly a century later in the surroundings of CERN that John Bell, an
alumnus of another Queen's on the island of Ireland, showed that the laws of
quantum theory predict empirical models that exhibit what we here call
contextuality. This is the surprising content of the celebrated Bell
\cite{bell64} and Bell--Kochen--Specker \cite{bell66,ks} theorems. By now it
has been confirmed that the measurement statistics obtained in a variety of
experiments witness a violation of Boole's conditions, recast and renamed as
Bell inequalities: \eg
\cite{aspect1982experimental,kirchmair2009state,hensen2015loophole,shalm2015strong,giustina2015significant}.

A strong result connecting Boole to Bell, proved by Samson and Lucien Hardy,
shows that a complete set of inequalities characterising the polytope of
non-contextual empirical models can be derived from logical consistency
conditions \cite{abramskyhardy:logical}. This should strike the reader as
rather alarming: the observed behaviour of physical systems can apparently
satisfy an inconsistent set of logical formulae! The saving grace in this
drastic situation is that, since observations may only arise in context, one can
never observe enough events at once to manifest the inconsistency. In relation
to this Samson has often been heard to say that quantum systems skirt the
borders of logical paradox.\footnotemark

\footnotetext{This also brings to mind the words of \'Alvaro de
Campos, as if quantum systems were contriving to realise the motto from
his futurist phase, `to be sincere contradicting oneself'.
In the original: `Ser sincero contradizendo-se'.
From the poem \textit{Passagem das horas} (22-05-1916), in \cite{alvarodecampos}.}

It is worth stressing that contextuality is a property of the observable
behaviour of a system. It is independent of whatever theory, quantum or
otherwise, that might be conjectured to account for that behaviour. The power of
the aforementioned seminal theorems of quantum foundations is that
contextuality cannot simply be dismissed as some bizarre artifice of an
incomplete or inadequate mathematical formulation of a physical theory \cite{einstein1935can} -- it is
an unavoidable feature in \stress{any} theory that accounts for the empirical
behaviours that have been observed in experiments. This is what grants contextuality
its \stress{phenomenological} status, and what justifies its being tested
experimentally. No matter how skilled the sketch artists are, a picture of
the whole elephant will always elude them -- a true elephant in the room.

\subsection*{Laying foundations for quantum information}

While Bell's theorem laid bare this counter-intuitive aspect of quantum
behaviours, less conclusive clues about the theory's counter-intuitive nature
had already been noticed for some time and had been a source of philosophical
or interpretational unease for several of its founding contributors: \eg the
EPR paradox \cite{einstein1935can}, or the Schr\"odinger's cat
thought experiment \cite{schrodinger1935gegenwartige}.

But surprising, or
troubling, as quantum theory may be in this respect, perspectives on the matter have
broadened in recent decades.
It was only a matter of time before one of the
great slogans of programming and hacker culture came to be applied to
quantum theory too: `it's not a bug, it's a feature!'.
The idea is that 
the use of quantum systems to carry and manipulate information opens up the
possibility of exploiting their non-classical \stress{weirdness} to attain
advantage in computational or other information-processing tasks.

This perspective has led to a renewal of interest in foundational aspects of
quantum theory, as we strive for a systematic and effective understanding of
quantum advantage. To get there requires being able to reason about information
processing at the quantum level, with all the common-sense-defying possibilities
it offers.

Contextuality, as the archetypal non-classical feature of quantum
phenomenology, has drawn particular attention in this regard. This has led to
the development of general, structural frameworks for treating contextuality,
including the sheaf \cite{ab}, graph \cite{csw2014graphtheoretic},
hypergraph \cite{afls} and contextuality-by-default \cite{dzhafarov2014contextuality} approaches.\footnotemark\ These interrelated frameworks
go well beyond the \textit{ad hoc} analysis of particular instances of the
phenomenon or the hunt for `small proofs' of the Bell--Kochen--Specker theorem,
topics that had been the main preoccupation of `all hitherto exisiting'
literature on contextuality. Instead, the focus is on developing
a general theory that distils the essence of contextuality and reveals its
structural and compositional aspects. The framework introduced by Samson and
Adam Brandenburger \cite{ab}\footnotemark\ is a particularly potent distillate
that emphasises these aspects and characterises contextuality as obstructions
to the passage from local to global. The framework gives elegant expression to
this idea through the mathematics of sheaf theory.

\addtocounter{footnote}{-1}
\footnotetext{Within these frameworks, the phenomenon of non-locality as
discussed by Bell may be seen as a special case of contextuality that arises
in distributed or multi-party scenarios. Note that locality in Bell's sense
differs from our use of the term earlier in relation to local compatibility.}
\stepcounter{footnote}
\footnotetext{Subsequent developments are to be found in many papers including
in particular the local-consistency-versus-global-inconsistency picture in
\cite{contextualitycohomologyparadox}.}

These structural frameworks have provided the basis for a range of recent results
that establish links between contextuality and quantum advantage
\cite{raussendorf2013contextuality,howard2014contextuality,abramsky2017contextual,bermejo2017contextuality,abramsky2017quantum,mansfield2018quantum,karanjai2018contextuality},
which have prompted further investigation into the r\^ole of contextuality as a
resource.

\subsection*{Resources: from objects to transformations}

In previous work with Samson we built upon the sheaf-theoretic framework to
develop a compositional resource theory of contextuality
\cite{abramsky2017contextual,karvonen2018categories,ourlics:comonadicview} (cf.
\cite{amaral2017noncontextual,amaral:resourcetheory}). The central notion is
that of simulation between empirical behaviours, which rests on an underlying
notion of experimental procedure between black boxes.

The scenario to have in mind is the following. Imagine an agent who has access
to a black box of the kind described before. They can perform experiments by
interacting with the box. A recipe for such an experiment specifies the agent's
actions surrounding their interaction with the box.

As we are only considering
single-use boxes, the instructions are of a rather simple form: they specify a
set $C$ of compatible measurements to be performed and a post-processing
function mapping the set of possible joint outcomes of these measurements to a
(new) set of outcomes for the experiment. For example, a recipe for an experiment
could read something
like `perform the compatible, dichotomic measurements $a$ and $b$
simultaneously, obtain their (Boolean) outcome values, and combine them with
\textsc{xor} to yield the (also dichotomic) result of the experiment'. One
could also include probabilistic mixtures of such deterministic experiments.

The agent may then package a collection of such experiments as a new (outer)
box, offering these derived measurements to external users through an
interface (Figure~\ref{fig:agent}). They must of course be careful to ensure that this external
interface only labels a set of (new) measurements as compatible if the agent is able
to perform all the corresponding experiments in parallel without running into
the compatibility limitations of the original black box.

In a slightly
synecdochical abuse of terminology, we call such a collection of instructions
for experiments an \emph{experimental procedure} (or just procedure, for short).
It is a procedure for using
the original box in building the new, wrapper box. Through such a procedure,
any behaviour (\ie empirical model) of the original box is converted into a
behaviour of the new box. However, it should be kept in mind that not every
possible behaviour of the new type is necessarily achievable in such a manner.

Setting the physics language aside for a moment, the kind of situation just described
will be very familiar to computer scientists. It arises all the time, for example, in
modular programming: one uses the functions or methods provided through an API
by a library (or by an object in object-oriented programming), whose implementation might be
hidden from us, in order to implement new functions or methods, which may in
turn be packaged as a new program module or library (or as a `wrapper' object) and
provided to other programmers.

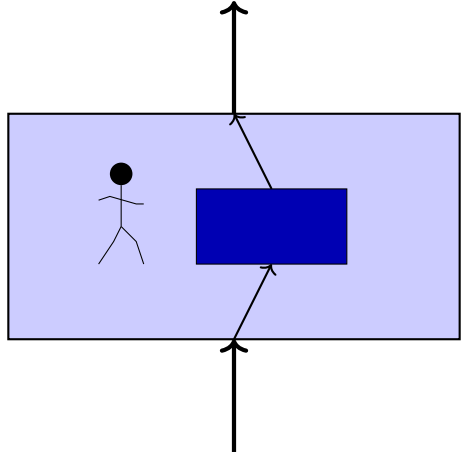
\begin{figure}
\begin{center}
\begin{tikzpicture}
\fill[fill=blue!20!white, draw=black,thick] (-3,-1.5) -- (3,-1.5) -- (3,1.5) -- (-3,1.5) --cycle; 
\fill[fill=blue!70!black, draw=black] (-.5,-.5) -- (1.5,-.5) -- (1.5,.5) -- (-.5,.5) --cycle; 
\draw[->, ultra thick] (0,-3) -- (0,-1.5);
\draw[->, ultra thick] (0,1.5) -- (0,3);
\draw[->, thick] (0,-1.5) -- (.5,-.5);
\fill (-1.5,.7) circle  (.15) ;
\draw (-1.5,.55) -- (-1.5,0) ;
\draw  (-1.5,0) -- (-1.3,-.2) -- (-1.2,-.5);
\draw  (-1.5,0) -- (-1.6,-.2) -- (-1.8,-.5);
\draw  (-1.2,.3) -- (-1.3,.3) -- (-1.65,.4) -- (-1.8,.35);
\draw[<-, thick] (0,1.5) -- (0.5,.5);
\end{tikzpicture}
\caption{An experimental procedure uses a black box of one type to simulate a black box of
another type.
\label{fig:agent}
}
\end{center}
\end{figure}

\subsection*{Simulations: inside-out and outside-in}

There is an alternative way to think about the notion of procedure, approaching
it `from the outside in', rather than  `from the inside out'. The difference is
that between synthesis and analysis: instead of `$A$ is converted to $B$', one
says that `$B$ is simulated from $A$' or that `$B$ reduces to $A$'.

From the
perspective of someone external to the new box, for whom it is just a black-box
system, the procedure followed by an internal agent may be posited as a
(typically incomplete) explanatory device for the empirical behaviour of the
system. This could be helpful, for example, in reducing the box's behaviour to
the behaviour of another black box, perhaps one of a simpler kind or one that is
more familiar and well understood. Indeed, as the agent's actions are fully
classical (non-contextual),
they essentially describe a simulation of the empirical behaviour of
the outer box from the behaviour of a posited inner box.

It may be noted that the agent in this story plays a double r\^ole. In relation
to the original black box they play the r\^ole of a user or of the
environment -- they may choose measurements and then obtain the respective outcomes.
But in relation to the wider world outside of the new box, they play the
r\^ole of a system -- they are prompted with measurement requests and must
produce an outcome.\footnotemark

\footnotetext{One cannot help but be reminded of the reversal of Player
and Opponent r\^oles in games of function type in game semantics. Player in
such a game plays simultaneously, and inter-dependently, two simpler games,
corresponding to the output and the input types, and adopts a different r\^ole
in each of them \cite{abramsky1999game}.}

\subsection*{To what end?}

In summary, taking the perspective of resource theory, the emphasis is no
longer placed on the individual black box behaviours but on simulations between
different instances of such behaviours. There are a number of reasons that
recommend this approach in the study of contextuality.

\begin{itemize}

  \item At least implicitly or informally, the notion of simulation has been
    central to a number of results in the non-locality and contextuality
    literature, \eg \cite{barrettpironio2005,barrettetal2005,jonesmasanes2005interconversions,dupuis2007nouniversalbox,allcock2009closedsets,forsterwolf2011bipartite}.
    A more explicit and structural formalisation of the concept can be useful
    for proving further results of this kind. Good examples of this are the
    no-copying result in \cite[Theorem 22]{ourlics:comonadicview} and the more general no-catalysis result in \cite{karvonen2021neither}.
    
  \item More broadly, resource theories provide a versatile setting that has
    already proved useful in exploring a variety of other resources, such as
    entanglement, in the field of quantum information, \eg
    \cite{horodecki2013quantumness,chitambar2019resource}.
    A more general
    mathematical framework encompassing these and many other examples can be
    found in \cite{coecke2016mathematical,fritz2017resource}. As with
    classical notions of reducibility, the existence of a simulation between
    one behaviour and another provides a way of comparing their degrees of
    contextuality. The induced preorder is richer, more expressive, and
    provides more structural and fine-grained distinctions than the linear
    order induced by a `measure of contextuality' such as the contextual
    fraction or others \cite{abramsky2017contextual,grudka2014quantifying}.
    In fact, there is not just one but a hierarchy of such preorders,
    analogous to the Abramsky--Brandenburger hierarchy of probabilistic,
    logical, and strong contextuality \cite{ab}. These preorders are determined by how
    flexible a notion of simulation one one wishes to consider. Ultimately,
    it is these preorders that should be regarded as being the fundamental
    concepts, while the various `measures' (or at least the linear orders
    they induce) are just somewhat arbitrary linear extensions of them.

    \item The resource-theoretic framework of simulations provides a unifying
      language in which several important concepts from the contextuality
      literature find common expression. The most striking examples are
      provided by two `corner cases': non-contextual models and Bell
      inequalities.
    
    \begin{itemize}
    
        \item Within our framework, non-contextual behaviours are naturally
          characterised as those that can be simulated from `nothing', \ie
          from the unique model on the empty black box that allows for no
          measurements. More precisely, non-contextual models on any given box
          are in one-to-one correspondence with (probabilistic) procedures from
          the empty box.\footnotemark
        
        \footnotetext{Note that in
        \cite{ourlics:comonadicview,abramsky2019simulations}, where we first
        made an observation to this effect, the source of these simulations was
        the trivial scenario with one measurement and a single outcome. The
        difference arises due to the kind of simulations we allow in each case.
        It is related to the fact that `the' singleton set is the terminal object
        in the category of sets and functions (\ie there is exactly one
        function from any given set to a singleton set) whereas in the
        category of sets and relations the terminal object is the empty set.}
        
        \item A novel, and perhaps surprising, fact that we explore in
          this chapter is that the well-studied notions of Bell inequalities
          and non-local games, suitably generalised to apply not just to
          non-locality but to contextuality more broadly, can also be described
          as experimental procedures.
          The key r\^ole here is played by the box `$\dice{2}$' that admits a single measurement
          with binary outcome. Its possible behaviours are thus
          characterised by a single number in the unit interval, specifying the
          probability that the outcome is $1$.
          A procedure from a given box to the box $\dice{2}$
          corresponds to a (normalised) Bell functional (the `left-hand side' of a Bell or non-contextuality inequality) 
          or to a non-local or contextual game.
          An empirical model on the given box is mapped through such a procedure
          to a number in $[0,1]$ -- this is the value of the functional or the
          winning probability of the game.
        
    \end{itemize}

    \item Last but not least, the shift in emphasis from single boxes to
      transformations between them is very much in the spirit of category
      theory. Category theory stems from the recognition that mathematical
      objects, to reappropriate and paraphrase the words of John
      Donne\footnote{From \textit{Meditation XVII}, in \cite{donne}.}, are best
      studied not as islands entire of themselves but as pieces of the
      continent, parts of the main. The same reasoning can be applied to
      `systems', physical or otherwise. In other words, what matters is not so
      much what things are, but how they stand in relation to one another. This
      perspective has proved to be immensely fruitful not only in mathematics
      but also in computer science, as a way to systematise concepts, recognise
      structure, and build bridges between \textit{a priori} disparate
      subjects. Having a category -- of boxes and procedures, or of behaviours
      and simulations, in this case -- provides us with a powerful tool to
      think about the structure of the theory at hand. As an example, the
      results in this chapter were in a sense motivated from categorical
      considerations -- but nevertheless they may still be stated more
      concretely, without the high-flown language.
    
\end{itemize}

\subsection*{Summary of results}

\subsubsection*{A question: from objects to maps}

We now turn to an overview of the main results in this chapter. A
procedure $f$ that uses a box of type $S$ to build one of type $T$ determines a
\stress{state transformation}, a function $\fdec{\EMP(f)}{\EMP(S)}{\EMP(T)}$ mapping any
empirical model for $S$ to one for $T$. It is natural to ask which functions
arise in this way. We are thus led to the following question:
\begin{equation}\label{q:proc}\tag{A}
\begin{minipage}{0.8\textwidth}
\textit{Given a function $\fdec{F}{\EMP(S)}{\EMP(T)}$, can it be realised by an
experimental procedure? \Ie is there a procedure $\fdec{f}{S}{T}$ such that
$F = \EMP(f)$?}
\end{minipage}
\end{equation}
In light of the remarks made above, when $S$ is the empty box, its set of empirical models $\EMP(S)$ is a
singleton. So, a function $\EMP(S) \to \EMP(T)$ simply picks out an element of
$\EMP(T)$. The question then reduces to a more familiar one:
\begin{equation}\label{q:ncmodel}\tag{B}
\begin{minipage}{0.8\textwidth}
    \textit{Given an empirical model, is it non-contextual?}
\end{minipage}
\end{equation}

This latter question of detecting contextuality of a given empirical model,
\ie from the observable probabilistic behaviour of a black box, has been
extensively studied and is by now well understood.
It can be phrased as a linear programming problem: the noncontextual empirical models form a convex polytope, whose vertices are the deterministic models corresponding to global assignments of outcomes to all measurements, and the question is to determine whether a given empirical model belongs to the polytope.

The more general question
\eqref{q:proc} can be seen as a relativised version of \eqref{q:ncmodel}.
The generalisation is somewhat in the spirit of the `relative point of view'  \cite{nlab:relativepov},
advocated by Grothendieck in the context of
algebraic geometry. The idea is that the consideration of properties of objects
gives way to the consideration of properties relativised to morphisms. This
results in fact in a strict generalisation, as properties of objects can be
then viewed as properties of the unique morphism from that object to the
terminal object.

\subsubsection*{An answer: from maps to objects}

We provide an answer to question \eqref{q:proc} in the form of necessary and
sufficient conditions for a map between empirical models to be realisable by an
experimental procedure (Theorem~\ref{thm:mainthm}). We now give a brief
summary of the main ingredients.

As a preliminary observation, note that the set $\EMP(S)$ of empirical models,
or possible behaviours, of any given box $S$ is a convex set. In other words, it is closed
under probabilistic mixtures: for any two conceivable behaviours, a mixture of
them is also a behaviour. Operationally, one can think that whoever prepares
the single-use copies of the black box does so by first throwing a (biased)
coin and then preparing it to behave according to one or the other. Moreover,
any behaviour transformation that is induced by a procedure preserves convex
combinations. We may thus restrict attention, in answering question
\eqref{q:proc}, to functions $\fdec{F}{\EMP(S)}{\EMP(T)}$ with this property.

Past this initial hurdle, we get to the crux of our characterisation. The short
summary is that \eqref{q:proc} is answered by reducing it back to
a question closely related to
\eqref{q:ncmodel}.
This is achieved by representing behaviour transformations as behaviours themselves.
Modulo a few details, it goes as follows.
For any boxes $S$ and $T$, we build a new box of (function) type, $[S,T]$.
Its possible behaviours, the empirical models in $\EMP([S,T])$,
induce convex-combination-preserving transformations 
$\EMP(S) \to \EMP(T)$.
Those transformations turn out to be realisable as procedures $S \to T$ precisely when they are induced by non-contextual empirical models.

In order to make this work, however, we need a refinement of our notion
of box. It consists of an added specification that restricts the allowed
behaviours of the box, somewhat akin to a type or class invariant
in programming. One may think of it as a contracted promise that comes
attached to the box interface and which any behaviour must fulfil. Concretely,
this is given as a \stress{predicate}, implemented as a two-valued
experiment to which the allowed behaviours are pledged to always return the
outcome $1$. In terms familiar in the quantum foundations and information literature, the
behaviours must be perfect strategies for a given (non-local or contextual) game.
So, the construction of the box $[S,T]$ also involves specifying such a
predicate, $g_{S,T}$.

The box $[S,T]$ has the following interface: measurements (and their compatibility structure)
correspond to those of $T$, while outcomes specify procedures for interacting
with the box $S$ in order to obtain an outcome in $O_T$ (the outcome set
of $T$). The r\^ole of
the predicate $g_{S,T}$ is to ensure that these procedures never lead to an
invalid use of $S$. In particular, following any set of procedures obtained as joint
outcomes from the box $[S,T]$ should only ever require measuring a
context of compatible measurements in $S$.

Thus, an empirical model of this box determines a map transforming empirical models of $S$ to empirical models of $T$.
The full answer to \eqref{q:proc} is then stated as follows:
a map $\fdec{F}{\EMP(S)}{\EMP(T)}$ that preserves convex combinations
is realised by a procedure $S \to T$
if and only if 
it is induced by an empirical model of $[S,T]$ that
satisfies the predicate $g_{S,T}$ and is non-contextual.
The question of whether a convex-combination-preserving map is realised by a procedure is thus that of membership in a certain image of the polytope of noncontextual models, also a linear programming problem.
Moreover, for the case of deterministic procedures, we obtain a sharper characterisation in Theorem~\ref{thm:characterizingdeterministicprocedures}:
roughly speaking, we probe the map $\fdec{F}{\EMP(S)}{\EMP(T)}$ along measurement contexts in $T$, and $F$ is realized by a deterministic procedure if and only if each of those composites is already determined by a measurement context in $S$.

The construction outlined above suggests that one regard the pair
$\tuple{[S,T],g_{S,T}}$ as somewhat akin to a function space, or more accurately a
space of procedures between black boxes. And indeed, if one considers the
category whose objects are boxes with invariant specifications, then this
construction provides it with a closed structure
\cite{eilenbergkelly:closedcats,laplaza:closedcats,Street:cosmoi,manzyuk:closedcats} (Theorem~\ref{thm:closure}).
This is an internalisation of the notion of hom-set, whereby the collection of
morphisms between any two objects is in a precise sense represented as an object in the
category itself.\footnotemark

\footnotetext{At the risk of overstretching the use of poetic metaphor, one
is reminded of Blake's, `[to] hold infinity in the palm of your hand, and
eternity in an hour.' From the poem `Auguries of Innocence' (c. 1803), in The Ballads (or Pickering) Manuscript, published in \cite{blake}.}

\section{The framework: objects}\label{sec:framework-objects}

We now introduce more carefully the framework that will be used throughout this chapter.
The basic setting that we outline in this section is that of the sheaf-theoretic approach to contextuality introduced by Abramsky \& Brandenburger \cite{ab}. The formalism provides general notions of measurement scenarios and empirical models,
respectively the \stress{types} or \stress{interfaces} and the \stress{states} or \stress{behaviours} of our black boxes.
In addition, a central r\^ole will be played by morphisms between these basic objects,
which were introduced in our previous work to underlie a resource theory of contextuality \cite{karvonen2018categories, abramsky2017contextual, ourlics:comonadicview},\footnotemark\
on which we elaborate further in the next section.

\footnotetext{In fact, the notion of morphism considered here differs slightly from those of \cite{karvonen2018categories} and \cite{ourlics:comonadicview}. It is the appropriate notion to capture \stress{non-adaptive} simulations. The minor discrepancies are discussed and explained in Remark~\ref{rem:whynotadaptive} at the end of Section~\ref{ssec:simulations}.}

\subsection{Measurement scenarios}

The first ingredient is the notion of measurement scenario.
This is a specification of the interface of a black box:
which queries (measurements) are available to the agent and
which type of response (outcome) they can elicit.
As such, measurement scenarios serve as the \stress{types}
(in the computer science sense) in our framework.

As discussed in the introduction,
the compatibility structure of measurements plays
a central r\^ole.
The scenario specifies which subsets of measurements form
contexts and can thus be jointly performed.
Note that if $C$ is a context, then any subset $U$ of $C$ must also be a context: 
the agent might as well jointly perform all the measurements in $C$
and disregard the outcomes of those in $C \setminus U$.
Moreover, it must be possible to perform any individual measurement
on its own.\footnotemark\
These two desiderata are neatly encapsulated in the notion of
abstract simplicial complex,
which is important in algebraic topology and combinatorics.

\footnotetext{Otherwise it would be difficult to justify calling it a measurement in the first place.}

 \begin{definition}
 An (abstract) \emph{simplicial complex} $\Sigma$ on a set of vertices $X$ is a family of finite subsets of $X$, called faces,
 which is non-empty, downwards-closed, and contains all the singletons. That is:
\begin{itemize}
    \item $\emptyset \in \Sigma$;
    \item for all $x \in X$, $\enset{x} \in \Sigma$;
    \item if $\sigma\in\Sigma$ and $\sigma' \subseteq \sigma$, then $\sigma' \in \Sigma$.
\end{itemize}
\end{definition}

Besides the measurements and their compatibility structure,
a measurement scenario must also specify the available outcomes for each measurement.
Putting these ingredients together, we arrive at the following definition.

\begin{definition}
A \emph{measurement scenario} is a triple $S=\scen{S}$ where:
\begin{itemize}
\item $X_S$ is a finite set ,whose elements are called \emph{measurements};
\item $O_S = (O_{S,x})_{x \in X_S}$ is a family that specifies, for each measurement $x \in X_S$, a finite non-empty set $O_{S,x}$, whose elements are the \emph{outcomes} of $x$;
\item $\Sigma_S$ is a simplicial complex on $X_S$, whose faces are called the \emph{measurement contexts}.
\end{itemize}
\end{definition}

We now introduce a couple of simple examples that will be useful for illustrating the main concepts.

\begin{example}\label{ex:trianglescen}
The (Auld) \emph{Triangle} scenario $\triangle$ has three measurements (or queries),
\[ X_{\triangle} \;\defeq\; \{\, \textit{`pint?'}, \quad \textit{`wine?'}, \quad \textit{`grub?'} \,\} \Mdot\]
The contexts admit no more than a pair of measurements to be performed at once; \ie the maximal faces of $\Sigma_\triangle$ are
\[ \{\textit{`pint?'},\; \textit{`wine?'}\}, \quad \{\textit{`wine?'},\; \textit{`grub?'}\}, \quad \{\textit{`pint?'},\; \textit{`grub?'}\} \Mdot\]
The outcomes (or responses) to each measurement $x \in X_\triangle$
take values in a two-element set,
\[ O_{\triangle,x} \;\defeq\; \{\, \textit{`yes'}, \quad \textit{`no'}\,\} \Mdot \]

This scenario can be found, albeit with different labels
for measurements and outcomes, in many other articles. It is the simplest
scenario in which contextuality can arise. In particular, it is
the scenario for Specker's parable of the overprotective seer \cite{liang2011specker}.
\end{example}

\begin{example}\label{ex:fourcandles}
The \emph{Four Candles} scenario $\square$ has four
measurements (queries),
\[ X_{\square} \;\defeq\; \{\, \textit{`SammyA'}, \quad \textit{`GeorgieB'}, \quad \textit{`JohnnyB'}, \quad \textit{`EvilG'} \,\} \Mdot\]
The maximal contexts, \ie the maximal faces of $\Sigma_\square$, are
\[
\{\textit{`SammyA'},\; \textit{`GeorgieB'}\}, \quad \{\textit{`GeorgieB'},\; \textit{`JohnnyB'}\}, \quad \{\textit{`JohnnyB'},\; \textit{`EvilG'}\}, \quad \{\textit{`SammyA'},\; \textit{`EvilG'}\} \Mdot
\]
The outcomes (or responses) to each measurement $x \in X_\square$
take values in a two-element set,
\[ O_{\square,x} \;\defeq\; \{\textit{`grain'}, \quad \textit{`grape'}\} \Mdot\]
We might think of the scenario as occurring when an agent needs to buy a
round of drinks for his four friends, but only gets to interrupt the conversation
enough to ask neighbouring pairs what they would like.

This scenario, again with different measurement and outcome
labels, is better known in the physics literature as the CHSH scenario. It is the scenario concerned by
perhaps the most well-known of Bell inequalities, originally formulated by Clauser, Horne,
Shimony, and Holt \cite{clauser1969proposed}.
\end{example}

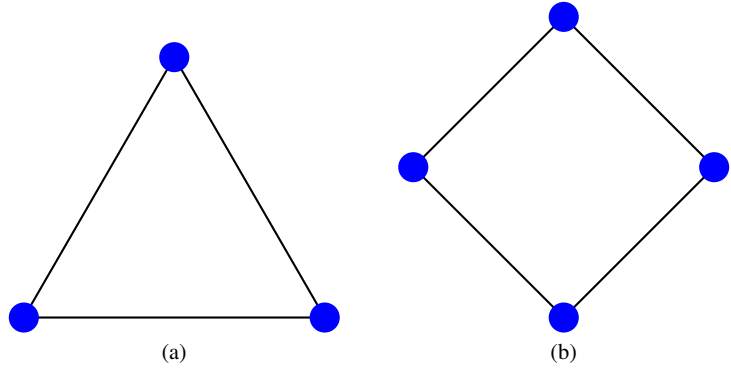
\begin{figure}[tbph]
\begin{center}
\begin{tabular}{cc}
\begin{tikzpicture}[scale=2]
  \path (-1, -1 ) coordinate (A1);
  \path ( 1, -1 ) coordinate (A2);
  \path ( 0, .732 ) coordinate (B1);

  \begin{scope}[ thick,black]
    \draw (A1) -- (B1) -- (A2) -- cycle;
  \end{scope}
 
  \fill[blue] (A1) circle(.1) ;
  \fill[blue] (A2) circle(.1) ;
  \fill[blue] (B1) circle(.1) ;
\end{tikzpicture}
\qquad & \qquad
\begin{tikzpicture}[scale=2]
  \path (-1, 0 ) coordinate (A1);
  \path ( 1, 0 ) coordinate (A2);
  \path ( 0,-1 ) coordinate (B1);
  \path ( 0, 1 ) coordinate (B2);

  \begin{scope}[ thick,black]
    \draw (B1) -- (A2) -- (B2) -- (A1) -- (B1);
  \end{scope}
 
  \fill[blue] (A1) circle(.1) ;
  \fill[blue] (A2) circle(.1) ;
  \fill[blue] (B1) circle(.1) ;
  \fill[blue] (B2) circle(.1) ;
\end{tikzpicture}
\\ (a) \qquad & \qquad (b)
\end{tabular}
\caption{Simplicial complexes representing measurement compatibility in (a) the
$\triangle$ scenario and (b) the $\square$ scenario.}
\end{center}
\end{figure}

We now introduce an extremely simple family of measurement scenarios
that will play a significant r\^ole in our treatment to follow.

\begin{example}\label{ex:nn}
For each positive integer $n$,
we denote by $\nn$ the scenario that
has a single measurement with $n$ possible outcomes, \ie
\[
X_\nn \,\defeq\, \enset{*}
\Mand
O_{\nn,*} \,\defeq\, \enset{0,\ldots, n-1}
\Mdot
\]
Note that $\Sigma_\nn$ is uniquely determined.
\end{example}

\subsection{Empirical models}

Having formalised how to describe the interface of the black boxes of interest, we now turn our attention to their possible behaviours.
These are described by empirical models on the given measurement scenario.

We start by considering the atomic empirical events
that represent a single interaction with a black box.

\begin{definition}
Let $S$ be a scenario.
For any $U \subseteq X_S$, we write \[\Ev_S(U) \defeq \prod_{x \in U} O_x \]
for the set of assignments of outcomes to each measurement in the set $U$.
When $U$ is a valid context, these are the joint outcomes one might obtain for the measurements in $U$.

The mapping above extends to a sheaf  $\fdec{\Ev_S}{\mathcal{P}(X_S)\op}{\cat{Set}}$,
called the \emph{event sheaf},
with restriction maps \[\fdec{\Ev_S({U\subseteq V})}{\Ev_S(V)}{\Ev_S(U)}\] given by the obvious projections.
When it does not give rise to ambiguity,
we often omit the subscript and denote the event sheaf more simply by $\Ev$.
\end{definition}

An account of the black box's behaviour ought to specify
its response to any allowed interaction (measurement context),
as a probability distribution over corresponding atomic events.

Let $\Dist$ denote the (discrete) probability distribution functor.
That is, for any set $X$,
$\Dist(X)$ is the set of (finitely-supported) probability distributions on $X$,
and for any function $\fdec{f}{X}{Y}$,
$\Dist(f)$ maps a distributions $d$ on $X$
to its push-forward along $f$, a distribution on $Y$.
Similarly, we denote by
$\Dist_\BB$ the Boolean distribution functor;
\ie $\Dist_\BB(X)$ is the set of Boolean-valued distributions on $X$.
Note that this is the covariant non-empty powerset functor.\footnote{Note that for Boolean distributions the restriction to finite support is unnecessary. But here we only deal with finite sets of events, anyway.}

\begin{definition}
A (probabilistic) \emph{empirical model} $e$ on a scenario $S$, written $e \colon S$,
is a compatible family for $\Sigma$ on the presheaf $\Dist \circ \Ev_S$.
More explicitly, it is a family $\family{e_\sigma}_{\sigma \in \Sigma_S}$ where,
for each $\sigma\in\Sigma_S$,
\[e_\sigma \;\in\; \Dist \circ \Ev(\sigma) \,=\, \Dist\left( \prod_{x \in \sigma}O_x \right)\]
is a probability distribution over the joint outcomes for the measurements in the context $\sigma$.
These distributions are required to be compatible in the sense that
for any $\sigma, \tau \in \Sigma_S$ with
$\tau \subseteq \sigma$,
one must have 
$e_{\tau} = e_{\sigma}|_\tau$,
where \[e_{\sigma}|_\tau \,\defeq\, \Dist \circ \Ev(\tau \subseteq \sigma)(e_\sigma) \]
is the marginalisation of $e_\sigma$ to the smaller context $\tau$;
\ie for any $t \in \Ev(\tau)$,
\[e_\tau(t) \,\defeq \sum_{s \in \Ev(\sigma), s|_\tau = t} e_\sigma(s)\Mdot\]
The set of all probabilistic empirical models on $S$ is denoted $\EMP(S)$.
\end{definition}

Note that compatibility can equivalently be expressed as the requirement that for all facets (\ie maximal contexts) $C$ and $C'$ of $\Sigma$,
\[e_C|_{C\cap C'} = e_{C'}|_{C\cap C'} \Mdot\]
Compatibility holds for all quantum-realisable behaviours \cite{ab}.
It generalises a property known as \emph{no-signalling} \cite{ghirardi1980general}, as will be illustrated shortly in Example~\ref{ex:ns}.

\begin{example}\label{ex:trianglemodel}
The following table provides an example of an empirical model in the $\triangle$ scenario.
The rows of the table specify probability distributions for the maximal contexts, and all
other distributions can simply be obtained by marginalisation.
\[
\begin{array}{ll||c|c|c|c}
  ~ & ~ & \textit{`yes'}\,\textit{`yes'} & \textit{`yes'}\,\textit{`no'} & \textit{`no'}\,\textit{`yes'} & \textit{`no'}\,\textit{`no'} \\
  \hline\hline
  \textit{`pint?'}    & \textit{`wine?'} & 0 & \sfrac{1}{2} & \sfrac{1}{2} & 0 \\
  \textit{`wine?'} & \textit{`grub?'}    & \sfrac{1}{2} & 0 & 0 & \sfrac{1}{2} \\
  \textit{`pint?'}    & \textit{`grub?'}    & \sfrac{1}{2} & 0 & 0 & \sfrac{1}{2} \\
\end{array}
\]
This is a black box behaviour that solves the following conundrum: 
one might wish to have a glass of wine if and only one is also eating; ditto for a pint of beer. However, one might want to have either beer or wine but not both.
An empirical model realising this scenario permits all of these constraints to be satisfied as long as no more than two questions are asked at once.
\end{example}

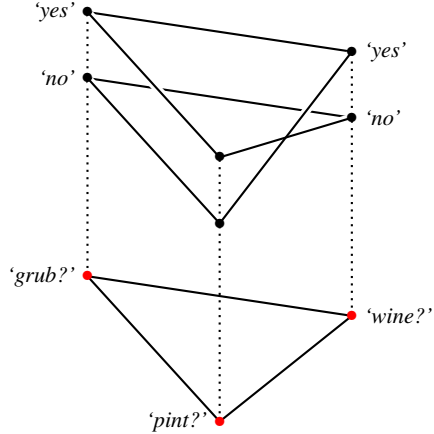
\begin{figure}[tbph]
    \centering
\begin{tikzpicture}[scale=1,x=50pt,y=50pt,thick,label distance=-0.25em,baseline=(current bounding box.center)]
\node (e) at (2,-0.3) {};
\node (n) at (1,0.8) {};
\node (T) at (0,1.5) {};
\node (u) at (0,0.5) {};
\node (uu) at (0,0.3) {};
\node [inner sep=0em] (a) at (0,0) {};
\node [inner sep=0em] (b) at ($ (a) + (e) $) {};
\node [inner sep=0em] (a') at ($ (a) + (e) + (n) $) {};
\node [inner sep=0em] (b') at ($ (a) + (n) $) {};
\node [inner sep=0em] (a1) at ($ (a) + (T) $) {};
\node [inner sep=0em] (a0) at ($ (a1) + (u) $) {};
\node [inner sep=0em] (b0) at ($ (b) + (T) $) {};
\node [inner sep=0em] (b2) at ($ (b0) + (u) $) {};
\node [inner sep=0em] (a'1) at ($ (a') + (T) $) {};
\node [inner sep=0em] (a'0) at ($ (a'1) + (u) $) {};
\node [inner sep=0em] (b'1) at ($ (b') + (T) $) {};
\node [inner sep=0em] (b'0) at ($ (b'1) + (u) $) {};

\draw (a') -- (b) -- (b') -- (a');

\draw [dotted] (b2) -- (b);

\draw [dotted] (a'0) -- (a');
\draw [dotted] (b'0) -- (b');

\node [inner sep=0.1em,label=left:{$\textit{`pint?'}$}] at (b) {{\color{red}$\bullet$}};
\node [inner sep=0.1em,label=right:{$\textit{`wine?'}$}] at (a') {{\color{red}$\bullet$}};
\node [inner sep=0.1em,label=left:{$\textit{`grub?'}$}] at (b') {{\color{red}$\bullet$}};

\draw [line width=3.2pt,white] (a'1) -- (b'1);
\draw [line width=3.2pt,white] (a'0) -- (b'0);
\draw (a'1) -- (b'1);
\draw (a'0) -- (b'0);
\draw [line width=3.2pt,white] (a'0) -- (b0);
\draw [line width=3.2pt,white] (a'1) -- (b2);
\draw (a'0) -- (b0);
\draw (a'1) -- (b2);
%
\draw [line width=3.2pt,white] (b2) -- (b'0);
\draw [line width=3.2pt,white] (b0) -- (b'1);
\draw (b2) -- (b'0);
\draw (b0) -- (b'1);

\node [inner sep=0.1em] at (b0) {$\bullet$};
\node [inner sep=0.1em] at (b2) {$\bullet$};

\node [inner sep=0.1em,label=right:{$\textit{`yes'}$}] at (a'0) {$\bullet$};
\node [inner sep=0.1em,label=right:{$\textit{`no'}$}] at (a'1) {$\bullet$};
\node [inner sep=0.1em,label=left:{$\textit{`yes'}$}] at (b'0) {$\bullet$};
\node [inner sep=0.1em,label=left:{$\textit{`no'}$}] at (b'1) {$\bullet$};
\end{tikzpicture}
    \caption{Bundle diagram for the empirical model of Example \ref{ex:trianglemodel}.}
    \label{fig:bundle}
\end{figure}

Another useful way to represent the above empirical model is as a bundle diagram
(Figure~\ref{fig:bundle}).
This method of representing models was introduced in \cite{contextualitycohomologyparadox}.
In the `downstairs' part of this diagram lives the simplicial complex $\Sigma_\triangle$ that captures measurement compatibility.
Above each measurement is a fibre consisting of its available outcome values.
Finally, the faces (just edges in this case, as the complex is $1$-dimensional)
in the `upstairs' of the diagram represent possible joint outcomes.
Note that while the empirical model associates probabilities with the joint outcomes,
the bundle representation in general carries less information,
as it simply includes faces for
those joint outcomes that have non-zero probability and omits the others.
As we will elaborate on shortly, in
many cases such (possibilistic) information about the support of the probability distributions already suffices to pick out interesting features of an empirical model.
As it happens, in this specific example,
the supports of the distributions uniquely determine the probabilities, by compatibility.

\begin{example}\label{ex:pr}
Similarly, the following table provides an example of an empirical model in the $\square$ scenario.
Here we notice that all context pairs give correlated outcomes except the context
consisting of $\textit{`SammyA'}$ and $\textit{`EvilG'}$, whose outcomes are anti-correlated.
\[
\begin{array}{ll||c|c|c|c}
  ~ & ~ & \textit{`grape'}\,\textit{`grape'} & \textit{`grape'}\,\textit{`grain'} & \textit{`grain'}\,\textit{`grape'} & \textit{`grain'}\,\textit{`grain'} \\
  \hline\hline
  \textit{`SammyA'} & \textit{`EvilG'}    & 0 & \sfrac{1}{2} & \sfrac{1}{2} & 0 \\
  \textit{`SammyA'} & \textit{`GeorgieB'} & \sfrac{1}{2} & 0 & 0 & \sfrac{1}{2} \\
  \textit{`JohnnyB'}& \textit{`EvilG'}    & \sfrac{1}{2} & 0 & 0 & \sfrac{1}{2} \\
  \textit{`JohnnyB'}& \textit{`GeorgieB'} & \sfrac{1}{2} & 0 & 0 & \sfrac{1}{2} \\
\end{array}
\]
\end{example}

\begin{example}\label{ex:ns}
The following table gives an empirical model on the $\square$ scenario,
whose measurements and outcomes have been relabelled.
This is the \emph{CHSH model}.
\[
\begin{array}{ll||c|c|c|c}
  ~ & ~ & 0\,0 & 0\,1 & 1\,0 & 1\,1 \\
  \hline\hline
  X_A & X_B    & \sfrac{1}{2} & 0 & 0 & \sfrac{1}{2} \\
  X_A & Y_B    & \sfrac{3}{8} & \sfrac{1}{8} & \sfrac{1}{8} & \sfrac{3}{8} \\
  Y_A & X_B    & \sfrac{3}{8} & \sfrac{1}{8} & \sfrac{1}{8} & \sfrac{3}{8} \\
  Y_A & Y_B    & \sfrac{3}{8} & \sfrac{1}{8} & \sfrac{1}{8} & \sfrac{3}{8} \\
\end{array}
\]
Here we can understand the black box as being shared by two parties $A$ and $B$;
each party can independently choose to make an $X$ or a $Y$ measurement, and for all measurements
the outcomes take values $0$ or $1$. Notice, for example,
that the probability that $X_A$ returns
$0$ is independent of which measurement party $B$
chooses to perform ($\sfrac{1}{2}+0 = \sfrac{1}{2}$ when $X_B$ is measured and $\sfrac{3}{8}+\sfrac{1}{8} = \sfrac{1}{2}$ when $Y_B$ is measured). This is the
no-signalling property at work.
It implies that the black box cannot be used by the parties to transmit information instantaneously to one another by means of their choice of measurement.
\end{example}

If we relabel and revisit Example~\ref{ex:pr}, viewing it as a bipartite
scenario as we have interpreted Example~\ref{ex:ns},
it corresponds to an empirical model known as the PR box~\cite{popescu1994quantum}.

\subsection{Possibilistic collapse}

As already mentioned, sometimes it is enough to consider the possibilistic information present in an empirical model.

\begin{definition}
A \emph{possibilistic} empirical model $e$ on a scenario $S$, also written $e \colon S$,
is a compatible family for $\Sigma$ on the presheaf $\Dist_\BB \circ \Ev_S$.
The set of possibilistic empirical models on $S$ is denoted $\EMP_\BB(S)$.
\end{definition}

Since Boolean distributions on a finite set can be identified with its non-empty subsets, one may describe a
possibilistic empirical model more explicitly as a family  
$\family{e_\sigma}_{\sigma \in \Sigma_S}$ where
each $e_\sigma$ is a non-empty subset of $\Ev(\sigma)$
and for any $\sigma, \tau \in \Sigma_S$ with $\tau \subseteq \sigma$,
\[e_{\tau} \,=\, e_{\sigma}|_\tau \,\defeq\, \setdef{s|_\tau}{s\in e_\sigma} \Mdot\]

The set $\EMP_\BB(S)$ of possibilistic empirical models on $S$ is equipped with a partial order, whereby $d\leq e$ 
if $d_\sigma\subseteq e_\sigma$ for all $\sigma\in\Sigma_S$.
The structure of this partial order was more thoroughly investigated in~\cite{abramsky2016possibilities}. We will make use of it when defining weak simulations in Definition~\ref{def:catsofempiricalmodels} below.

There is a natural map $\EMP(S)\to\EMP_\BB(S)$ which sends
a (probabilistic) model $e$ to its \emph{possibilistic collapse}.
This is the possibilistic model $e'$ defined by $e'_\sigma\defeq\supp(e_\sigma)$.
We will often abuse notation by denoting both a model and its possibilistic collapse by the same letter.
As an aside, note that not all models in $\EMP_\BB(S)$ arise via possibilistic collapse from a model in $\EMP(S)$ \cite{abramsky2013relational,abramsky2016possibilities}, a fact that has echoes of the main problematic set out in this chapter.

\begin{example}\label{ex:possibilisticns}
The following table gives the possibilistic collapse of the CHSH model.
\[
\begin{array}{ll||c|c|c|c}
  ~ & ~ & 0\,0 & 0\,1 & 1\,0 & 1\,1 \\
  \hline\hline
  X_A & X_B    & 1 & 0 & 0 & 1 \\
  X_A & Y_B    & 1 & 1 & 1 & 1 \\
  Y_A & X_B    & 1 & 1 & 1 & 1 \\
  Y_A & Y_B    & 1 & 1 & 1 & 1 \\
\end{array}
\]
\end{example}

\subsection{Contextuality}

We now come to the definition of contextuality.

\begin{definition} 
An empirical model $e \colon S$ is said to be {non-contextual} if it is extendable to a global section for ${\Dist \circ \Ev}$. In other words, $e$ is non-contextual if there exists a distribution $d \in \Dist \circ \Ev(X_S)$ on global assignments of outcomes to all measurements in $X_S$ such that $d |_\sigma = e_\sigma$ for every context $\sigma \in \Sigma_S$.
Otherwise, the empirical model is said to be \emph{contextual}.
\end{definition}

An equivalent formulation of non-contextuality of an empirical model $e$, which will be used throughout this chapter, is to say that $e$ can be written as a convex combination (taken contextwise) of deterministic empirical models. By a deterministic empirical model, we mean an empirical model $d \colon S$ in which $d_\sigma$ is a delta distribution for every context $\sigma \in \Sigma_S$, or equivalently, a compatible family for $\Sigma$ on the presheaf $\Ev$. Since $\Ev$ is in fact a sheaf, deterministic empirical models are in one-to-one correspondence with global assignments $s \in \Ev(X_S)$. We write $\delta_s$ for the deterministic model corresponding to the global assignment $s$.

As already mentioned, sometimes it is possible to witness contextuality at the possibilistic level, \ie by looking only at the supports of the empirical probability distributions.

\begin{definition}
An empirical model $e \colon S$ is said to be {logically non-contextual} if its possibilistic collapse is non-contextual over $\BB$, \ie if there is a Boolean distribution $d \in \Dist_\BB \circ \Ev(X_S)$ over joint outcomes that marginalises to (the possibilistic collapse of) $e$.
Otherwise, the empirical model is said to be \emph{logically contextual}.
\end{definition}

Note that in the case of logical contextuality there is a canonical candidate for the global distribution, namely that corresponding to the set
of global value assignments that are compatible with the supports of $e$,
\[\mathcal{S}_e \defeq \setdef{ s \in \Ev(X_S)}{\Forall{\sigma \in \Sigma_S} s|_C \in \supp e_C} \Mdot\]
That is, a model $e$ is logically non-contextual if and only if $\mathcal{S}_e|_{\sigma}$ is equal to the support of $e_\sigma$ for all contexts $\sigma \in \Sigma_S$. 
Equivalently, $e$ is logically contextual if and only if there is an
assignment $s \in \Ev_S(\sigma)$ in the support of $e_\sigma$ which cannot be extended to a global assignment in $\mathcal{S}_e$.

This leads us to an even stronger notion of contextuality.

\begin{definition}
An empirical model $e \colon S$ is said to be \emph{strongly contextual} if $\mathcal{S}_e = \emptyset$.
\end{definition}

In other words, $e$ is strongly contextual
if there is not even a single global assignment $s\in\Ev(X_S)$ that is compatible with the support of $e$ in every context.

These \stress{strengths} of contextuality form a strict hierarchy \cite{ab}. Strong contextuality implies
logical contextuality, which in turn implies probabilistic contextuality (for probabilistic empirical models).
Moreover, it is possible to find empirical models
that separate each of these classes.

The model from Example~\ref{ex:trianglemodel} is strongly contextual. In fact this
can be deduced simply by inspecting its bundle diagram in Figure~\ref{fig:bundle} and
noticing that any attempt at finding a univocal global value assignment consistent with
the observed outcomes must fail -- if one traces a path from jointly possible outcome
to jointly possible outcome it is impossible to close it while only ascribing a single
outcome value to each measurement.
Similarly the model from Example~\ref{ex:pr} is strongly contextual, a fact that can also
be deduced from a bundle diagram as was previously illustrated in
\cite{contextualitycohomologyparadox}. The CHSH model of Example~\ref{ex:ns}
is probabilistically contextual but it is not logically contextual. We refer the interested
reader to \cite{ab} for more details on the hierarchy of contextuality.
 
 \section{The framework: morphisms}\label{sec:framework-morphisms}
 
 We move to considering transformations between measurement scenarios.
 The various notions of morphism $S \to T$ introduced in this section 
 describe procedures that can be followed by a classical agent
 with access to a measurement scenario $S$ in order to implement the measurements
 of a new scenario $T$.
 Such procedures naturally induce simulations between empirical models on the scenarios at hand.
They can be regarded as the `free' operations in a (non-adaptive) resource theory of contextuality.

We discuss how some well-studied notions in the literature, notably non-local games, admit a neat description in terms of procedures and simulations.

We will generally refer to morphisms $S \to T$ as experimental procedures or just \emph{procedures},
with additional adjectives (deterministic, probabilistic, possibilistic) added as needed.
However, some particular instances of these concepts warrant specific terminology.
An \emph{experiment} on a scenario $S$ refers to 
a procedure of type $S\to \nn$, where $\nn$ is the scenario with a single measurement and $n$ possible outcomes from Example~\ref{ex:nn}.
When $n=2$, we also speak of a \emph{predicate}.\footnotemark

\footnotetext{While the word ``predicate'' does not quite fit with the experimental imagery evoked by much of our terminology, there is a reason for introducing an alternative word for two-valued experiments. Later, it will be useful to consider whether a given model always returns the outcome $1$ in such an experiment, and to restrict attention to those models on a scenario which do so. Hence calling it a predicate serves the purpose of indicating a change of viewpoint, where we will be restricting attention to models that always satisfy a given property.}

 \subsection{Deterministic procedures}\label{sec:classical}
 
 We first consider deterministic procedures.
 A procedure $S \to T$ consists of two parts:
 a (pre-processing) map of inputs in the backward direction and a (post-processing) map of outputs in the forward direction.
 More concretely, a procedure for implementing a measurement $x$ in $X_T$ must specify a context of measurements in $S$ to be performed and a way to map the outcomes obtained into values 
 in the outcome set of $x$. Importantly, it must be ensured that in the implementation of a context (\ie compatible set of measurements) of $T$, only a compatible set of measurements in $S$ is performed.
 This compatibility condition is captured by the concept of simplicial relation.
 
We start by fixing some notation.
Given a relation $\fdec{R}{X}{Y}$,
the image of an element $x \in X$ under $R$ is the set
\[
R(x) \defeq \setdef{y \in Y}{x R y}
\Mdot\]
Similarly, if $U$ is a subset of $X$ then
the image of $U$ under $R$ is the set
\[R(U) \defeq \bigcup_{x \in U}R(x) = \setdef{y \in Y}{\Exists{x \in \sigma} x R y}
\Mdot\]

\begin{definition}
Let $\Sigma$ and $\Delta$ be simplicial complexes.
A \emph{simplicial relation} $\fdec{R}{\Sigma}{\Delta}$
is a relation between the vertices of $\Sigma$ and those of $\Delta$
that maps faces to faces, \ie
such that for all $\sigma \in \Sigma$, $R(\sigma) \in \Delta$.
\end{definition}

\begin{definition}
A \emph{deterministic procedure} $\fdec{f}{S}{T}$ between measurement scenarios $S$ and $T$ is a pair $f = \tuple{\pi_f,\alpha_f}$ consisting of:
\begin{itemize}
    \item a simplicial relation $\fdec{\pi_f}{\Sigma_T}{\Sigma_S}$, which specifies for each measurement $x$ of $T$ a context $\pi_f(x)$ of $S$;
    \item a family $\alpha_f = \family{\alpha_{f,x}}_{x \in X_T}$ of functions $\fdec{\alpha_{f,x}}{\Ev_S(\pi_f(x))}{\Ev_T(x)}$, which map joint outcomes of $\pi_f(x)$ to outcomes of $x$.
\end{itemize}
The category of measurement scenarios and deterministic procedures is denoted by $\ScenDet$.
\end{definition}

\begin{remark}
Specifying a family of functions $\family{\fdec{\alpha_x}{\Ev_S(\pi(x))}{\Ev_T(x)}}_{x \in X_T}$ is equivalent to specifying a natural transformation $\fdec{\alpha}{\Ev_S(\pi -)}{\Ev_T(-)}$ of presheaves on $\mathcal{P}(X_T)$. This is because both of these presheaves are in fact sheaves (on $X_T$ with the discrete topology) and morphisms of sheaves can be glued together along any cover. 
\end{remark}

As remarked above, a \emph{deterministic experiment}
on a scenario $S$ is a deterministic procedure $S \to \nn$,
and a \emph{deterministic predicate} is a deterministic procedure $S \to \dice{2}$.

\begin{remark}
Note that a deterministic predicate on $S$ is given by a choice of $\sigma\in\Sigma_S$ and a subset of $\Ev_S(\sigma)$. 
\end{remark}

\begin{example}\label{ex:procsquarefromtriangle}
We consider an example of a deterministic experimental procedure $\fdec{f}{\triangle}{\square}$ where $\triangle$ and $\square$ are the scenarios from Examples~\ref{ex:trianglescen} and~\ref{ex:fourcandles} respectively. We define $\pi$ and $\alpha$ pointwise as follows. We send $\textit{`SammyA'}$ to query for $\textit{`pint?'}$ and to heed to the answer, \ie to choose $\textit{`grain'}$ if the answer is $\textit{`yes'}$ and $\textit{`grape'}$ otherwise. We send $\textit{`EvilG'}$ to query for $\textit{`wine?'}$ and instruct him to disobey the answer, \ie to go for $\textit{`grain'}$ if the answer is $\textit{`yes'}$ and $\textit{`grape'}$ otherwise. Finally, 
we send both $\textit{`GeorgieB'}$ and $\textit{`JohnnyB'}$ to the query $\textit{`grub?'}$, and instruct them to choose $\textit{`grain'}$, if the answer is $\textit{`yes'}$ and to choose $\textit{`grape'}$ otherwise. In effect, each party sends $\textit{`yes'}$ to $\textit{`grain'}$ and $\textit{`no'}$ to $\textit{`grape'}$; they just perform different queries on the triangle.
\end{example}

\subsection{Probabilistic and possibilistic procedures}\label{ssec:probabilisticprocs}

We now introduce a more general notion of morphism, which allows for some classical randomness.

\begin{definition}
A \emph{probabilistic procedure} $S\to T$ is a convex mixture of deterministic procedures $S\to T$, \ie an element of $\Dist(\ScenDet(S,T))$.
The category of measurement scenarios and probabilistic procedures is denoted by $\Scen$.
\end{definition}

We often denote such a probabilistic procedure as
$\sum_{i\in I} r_i f_i$ where $I$ is a (necessarily non-empty) finite set,
$r_i$ are positive reals summing to $1$,
and each $\fdec{f_i}{S}{T}$ is a deterministic procedure.

Similarly, when we care only about the possible events and not their specific probabilities,
we are led to a possibilistic version of procedures.
\begin{definition}
A \emph{possibilistic procedure} $S\to T$ is a Boolean mixture of deterministic procedures $S\to T$, \ie an element of  $\Dist_\BB(\ScenDet(S,T))$. 
The category of measurement scenarios and possibilistic  procedures is denoted by $\Scen_\BB$.
\end{definition}

We often write such a possibilistic procedure as
$\bigvee_{i\in I} f_i$ where $I$ is non-empty finite set and each $f_i$ is a deterministic procedure $S\to T$. 

Analogously to the deterministic case,
a \emph{probabilistic experiment} (resp. \emph{possibilistic experiment}) on a scenario $S$ is a
a probabilistic (resp. possibilistic) procedure $S \to \nn$,
and this is also called a \emph{probabilistic predicate} (resp. \emph{possibilistic predicate}) on $S$ when $n = 2$.

\subsection{Simulations}\label{ssec:simulations}

The morphisms of measurement scenarios introduced above yield maps between their sets of empirical models.
The idea is that by following a procedure $\fdec{f}{S}{T}$, an empirical model $e \colon S$ is used to simulate a new empirical model on $T$, denoted $\EMP(f)\,e$.

Mathematically, a helpful analogy is to think of empirical models as probability distributions on
a `contextual space', and to regard the simulation maps induced by (deterministic) procedures as analogous to the push-forward of a probability measure along a (measurable) function. Indeed, this is precisely how the map is defined at each context.

\begin{definition}\label{def:EMPprobabilistic}
Given a deterministic procedure $\fdec{f}{S}{T}$
we define a function $\fdec{\EMP(f)}{\EMP(S)}{\EMP(T)}$ by setting 
    \[(\EMP(f)\,e)_\sigma\defeq \alpha^*_{f,\sigma}(e_{\pi_f \sigma}) \Mcomma\]
\ie the probability distribution that $\EMP(f)\,e$ gives at context $\sigma\in\Sigma_T$
is obtained by pushing forward $e_{\pi_f\sigma}$ along the map $\fdec{\alpha_{f,\sigma}}{\Ev_S(\pi_f\sigma)}{\Ev_T(\sigma)}$.
This is extended to probabilistic procedures, which are convex mixtures $\sum_i r_i f_i$ of deterministic procedures, by defining \[\EMP\left(\sum_i r_i f_i\right) e \defeq \sum r_i (\EMP(f_i)e)\Mdot\]
\end{definition}

In this way, $\EMP$ defines a functor with domain $\Scen$. As for its codomain, note that $\EMP(S)$ has more structure than just that of a bare set. Namely, one can naturally form convex combinations $\sum_{i=1}^n r_i e_i$ of empirical models, and so $\EMP(S)$ has the structure of a convex set; \ie it is an algebra of the distribution monad on $\cat{Set}$. The following lemma shows that $\EMP(f)$ preserves this structure.

\begin{lemma}\label{lem:simsareconvex}
If $\fdec{f}{S}{T}$ is a probabilistic procedure, then $\EMP(f)$ preserves convex combinations.
\end{lemma}
\begin{proof}
We observe first that the claim holds whenever $f$ is a deterministic procedure. Indeed, as $\EMP(f)\,e$ is defined for $\sigma\in\Sigma_T$ by 
\[(\EMP(f)\,e)_\sigma=\alpha^*_{\sigma}(e_{\pi \sigma})\]
this follows from the fact that $\alpha^*_\sigma$ preserves convex combinations for each $\sigma$. Moreover, when $f=\sum r_if_i$ is a convex combination of deterministic procedures $f_i$ the function $\EMP(f)$ is defined via $\sum r_i \EMP(f_i)$, and so it follows that $\EMP(f)$ preserves convex combinations since each $\EMP(f_i)$ does.
\end{proof}

Therefore, we can think of $\EMP$ as a functor whose codomain is the category of convex sets (with convex-combination-preserving functions). However, at times we find it convenient to abuse notation and compose this with the forgetful functor to $\cat{Set}$ without comment.
 
 Formally speaking, the action of $\EMP$ on probabilistic procedures follows inevitably from its action on deterministic procedures. This is essentially because $\EMP$ is a convex set. The passage from $\ScenDet$ to $\Scen$ can be seen as freely enriching $\ScenDet$ in convex sets. As $\EMP$ already gives a functor from $\ScenDet$ to convex sets, it induces an enriched (\ie convex) functor from $\Scen$ to convex sets.

There is an similar construction for possibilistic procedures,
to which the analogous remarks apply (replacing convex sets by sup-lattices).
\begin{definition}
A deterministic procedure $f$ induces a function $\fdec{\EMP_\BB(f)}{\EMP_\BB(S)}{\EMP_\BB(T)}$
defined in the same way as for probabilistic empirical models (Definition~\ref{def:EMPprobabilistic}).
This is extended to possibilistic procedures by setting
\[\EMP_\BB\left(\bigvee_i f_i\right)e \defeq \bigvee_i (\EMP_\BB(f_i)e) \Mdot\]
\end{definition}

We now have all the necessary ingredients in place to introduce the relevant notions of simulation
between empirical models.

\begin{definition}\label{def:catsofempiricalmodels}
The category $\cat{Emp}$ of probabilistic empirical models and (probabilistic) simulations is defined as the category of elements of $\EMP$.
Explicitly, the objects of $\cat{Emp}$ are probabilistic empirical models $e\colon S$,
and morphisms $e\colon S \to d\colon T$ are \emph{probabilistic simulations}, \ie probabilistic procedures $\fdec{f}{S}{T}$ such that $\EMP(f)\,e=d$.

The category $\cat{Emp}_\BB$ of possibilistic empirical models and possibilistic simulations is defined as the category of elements of $\EMP_\BB$. Explicitly, the objects of $\cat{Emp}_\BB$ are possibilistic empirical models $e\colon S$, and  morphisms  $e\colon S\to d\colon T$ are \emph{possibilistic simulations}, \ie  possibilistic procedures $\fdec{f}{S}{T}$ such that $\EMP_\BB(f)\,e=d$.

The category $\cat{Emp}_\BB^{\leq}$ of possibilistic empirical models and weak simulations is defined as the lax category of elements of $\EMP_\BB$. Explicitly, the objects of $\cat{Emp}_\BB^{\leq}$ are possibilistic empirical models $e\colon S$, and  morphisms  $e\colon S\to d\colon T$ are given by \emph{weak simulations}, \ie  possibilistic procedures $\fdec{f}{S}{T}$ such that $\EMP(f)_\BB\,e \leq d$.
\end{definition}

While the names of the morphisms above are chosen to elicit a helpful intuition, let us explain this in some more detail. We recommend that one thinks of a morphism $e\to d$ as a way of simulating $d$ from $e$. The underlying morphism of scenarios gives an operational description of the simulation procedure, whereas the choice of category amounts to choosing a notion of \stress{correctness} for such a procedure.

For $\fdec{f}{S}{T}$ to define a probabilistic simulation $e\to f$ in $\cat{Emp}$, the statistics of $e$ need to be taken to the statistics of $d$ exactly. The adjective probabilistic refers to the fact that the procedure can be probabilistic. In other words, when simulating $d$ from $e$ one obtains exactly $d$.
In contrast to this, for $f$ to be a possibilistic simulation $d\to e$ in $\cat{Emp}_\BB$ it is sufficient that $\EMP(f)\,e$ and $d$ have the same support. Finally, for a morphism $S\to T$ to define a weak simulation $e\to d$, it is enough that one never observes an outcome that $d$ deems impossible when running the simulation using $e$.

While weaker and weaker forms of simulation may seem too weak to be of practical interest,
note that the corresponding notions of \emph{non-simulability} between empirical models become stronger as one relaxes the type of simulation considered.
With this in mind, we say that an empirical model $e$ is probabilistically / possibilistically / strongly non-simulable from $d$ if there are no morphisms $d\to e$ in $\cat{Emp}$ / $\cat{Emp}_\BB$ / $\cat{Emp}_\BB^{\leq}$. Theorem~\ref{thm:contextualiffnonsimulable} below says that probabilistic, possibilistic, and strong contextuality of a model $e$ correspond precisely to probabilistic, possibilistic, and strong non-simulability from a trivial model.

\begin{example}\label{ex:simPRfromtriangle}
A straightforward calculation shows that the procedure $\fdec{f}{\triangle}{\square}$ from Example~\ref{ex:procsquarefromtriangle} takes the model $e \colon \triangle$ from Example~\ref{ex:trianglemodel} to the model $d \colon \square$ from Example~\ref{ex:pr}, so that $f$ defines a simulation $e \to d$ in $\cat{Emp}$. 

One can in fact show that there is no simulation in the opposite direction. Indeed, consider an arbitrary simplicial relation $\fdec{\pi}{\Sigma_\triangle}{\Sigma_\square}$.
As any pair of measurements is compatible in the scenario $\triangle$, this must also be true in the image of $\pi$. This implies that this image must in fact be a face of $\Sigma_\square$.
Thus, a procedure $\fdec{f}{\square}{\triangle}$ can only make use of one maximal context of $\square$.
Note that the restriction to a context of any empirical model $d \colon \square$ is necessarily non-contextual. Consequently, so is $\EMP(\pi,\alpha)d$.
However, the model $e:\triangle$ from Example~\ref{ex:trianglemodel} is strongly contextual, 
and so it cannot be equal to $\EMP(f)\,d$ for any $d \colon \square$.
In fact, this model cannot even be weakly simulated from a model on $\square$, or for that matter from any quantum-realisable model.
This statement remains true even if one allows for probabilistic and adaptive procedures, but establishing that in detail is beyond the scope of the present chapter.
However, the crux of the argument is the same: in quantum-realisable scenarios a set of measurements is compatible if and only if it is pairwise compatible, a fact known as Specker's principle. Therefore, any attempt at simulating an empirical model on $\triangle$ from a model on such a quantum-realisable scenario ends up using only a fully compatible -- and thus non-contextual -- part.
\end{example}

Further examples of simulations can be found in the literature. For example, Protocols 1, 2, 5, and 6 in~\cite{barrettetal2005} define simulations between various empirical models.
The other protocols in that article go slightly beyond the present framework by being adaptive, approximately correct, or allowing for some limited classical communication.
Similarly, the proof of Corollary 2 in~\cite{barrettpironio2005}, which states that any two-output bipartite box can be simulated by sufficiently many PR boxes, does not require adaptivity, hence it also holds true in our current framework. 

\begin{remark}\label{rem:whynotadaptive}
The definitions of morphisms presented here differ slightly from those in earlier expositions in~\cite{karvonen2018categories} or~\cite{ourlics:comonadicview}. 
Our deterministic morphisms are the same as in~\cite{karvonen2018categories}, whereas in~\cite{ourlics:comonadicview} we required the basic morphisms to have an underlying simplicial map rather than a simplicial relation.
The bigger differences occur afterwards: in~\cite{karvonen2018categories} the expressive power of deterministic morphisms was increased by letting the components of $\alpha$ be stochastic maps. Our present notion of probabilistic procedure is more general in that it allows both $\pi$ and $\alpha$ to behave stochastically. The earlier limitation was just due to not seeing the current definition as a possibility. In contrast to this, \cite{ourlics:comonadicview} extends the deterministic morphisms differently, by passing to a coKleisli category of a certain comonad on $\Scen$ in order to allow \emph{adaptive} protocols instead of mere joint measurements as in our current setup.
Note that the presence of adaptivity elides the difference in expressive power between simplicial maps and simplicial relations: measuring a context with several measurements can be achieved through a measurement protocol that measures one measurement at a time (with a trivial form of adaptivity).
We would happily work in the same adaptive setting if only we knew how to generalise our current results, specifically  Theorem~\ref{thm:characterizingdeterministicprocedures}, to such adaptive morphisms.
As we currently don't know how to do this, we focus on the non-adaptive case and discuss the issues raised by adaptivity in Section~\ref{ssec:variantsofmainresult}.
A further difference is how randomness is dealt with: here we obtain shared, classical randomness by allowing convex mixtures of morphisms, whereas in~\cite{ourlics:comonadicview} we defined a morphism $d\to e$ to be a deterministic co-Kleisli morphism $d\otimes c\to e$ for some non-contextual $c$. In the presence of adaptivity, there is no difference in expressive power, but only a difference in viewpoint. The setup in~\cite{ourlics:comonadicview} suggests defining more general morphisms $d\to e$ as maps $d\otimes c\to e$ where the resource $c$ is in some fixed class of interest (\eg the quantum-realisable empirical models). A benefit of our current framework is that it makes it easier to discuss probabilistic morphisms at the level of scenarios already, rather than only at the level of empirical models. 
\end{remark}

\subsection{Contextuality as non-simulability}

The following result shows how the familiar notion(s) of contextuality can be neatly expressed in terms of simulations. As discussed below, this justifies the relative point of view, allowing in particular to regard our question~\eqref{q:proc} as a generalisation of the question of determining whether a model is contextual, our question~\eqref{q:ncmodel}.

\begin{theorem}\label{thm:contextualiffnonsimulable} Let $\Zero$ denote the unique scenario with an empty set of measurements, and let $\zero$ denote the unique empirical model on it.  A probabilistic model $e : S$ is contextual if and only if it is probabilistically non-simulable from $\zero$,
\ie if and only if there is no morphism $\zero \to e$ in $\cat{Emp}$.

Moreover, a possibilistic model $e$ is logically (resp. strongly) contextual if and only if it is possibilistically (resp. strongly) non-simulable from $\zero$, \ie if and only if there is no morphism $\zero\to e$ in $\cat{Emp}_\BB$ (resp. $\cat{Emp}_\BB^{\leq}$).
\end{theorem}

\begin{proof}
These characterisations of probabilistic and logical contextuality were already proved in~\cite[Theorem 4.1]{karvonen2018categories}, but we include a proof here for the sake of completeness.

The crucial observation is that deterministic procedures $\Zero \to S$ are in one-to-one correspondence with global assignments for $S$, \ie with elements of $\Ev_S(X_S)$, and thus with deterministic models on $S$. If $\fdec{f}{\Zero}{S}$ is a deterministic procedure, then $\pi_f$ is necessarily the empty relation, so the only choice is in choosing $\fdec{\alpha_{f,x}}{\Ev_\Zero(\emptyset)}{\Ev_S(x)}$ for each $x \in X_S$. Noting that  $\Ev_\Zero(\emptyset)$ is a singleton and $\Ev_S(x) = O_{S,x}$, we see that $\alpha_{f,x}$ is determined by an element of the outcome set of $x$, and so $\alpha_f$ precisely specifies a global assignment $s \in \Ev_S(X_S)$.
Writing $f_s$ for the deterministic procedure corresponding to $s$ in this way, note that $\EMP(f_s)\,z = \delta_s$, where $\delta_s$ is the deterministic model determined by $s$. So, in particular, every deterministic model is simulable from $\zero$.

We start with the probabilistic case and prove the implications in both directions by proving their contrapositives. If $e$ is non-contextual, then it can be expressed as a convex combination of deterministic models $e = \sum_{s \in \Ev(X_S)} r_s\delta_{s}$, and the procedure $\sum r_s f_{s}$ simulates $e$. Conversely, suppose $e$ is simulable from $\zero$, say by a probabilistic procedure $\fdec{f}{\Zero}{S}$. Then $f$ is a convex combination of deterministic procedures, which are of the form $f_s$ for a $s \in \Ev(X_S)$, and so  we can write $f = \sum_{s\in\Ev(X_S)} r_s f_s$. Consequently,
\[ e \;=\; \EMP(f) \, z \;=\; \sum_{s\in\Ev(X_S)} r_s\EMP(f_s)\,z \;=\; \sum_{s\in\Ev(X_S)} r_s \delta_s \Mcomma\]
thus $e$ is non-contextual.
The case of logical contextuality is obtained by replacing convex sums with possibilistic combinations in the above. 

We now move on to strong contextuality, again proving the two directions by showing their contrapositives. If $e$ is weakly simulable from $\zero$, the underlying procedure is necessarily of the form $\bigvee_{s \in A} f_s$ for some subset of global assignments $A \subseteq \Ev(X_S)$.
Then any of the deterministic models $\EMP_\BB(f_s)\,z = \delta_s$ for $s \in A$ is a witness of $e$ not being strongly contextual.
Conversely, if $e$ is not strongly contextual, there is a global assignment $s \in \Ev(X_S)$ consistent with it. The corresponding deterministic procedure $\fdec{f_s}{\Zero}{S}$ yields a possibilistic simulation $\zero \to \delta_s$ in $\cat{Emp}_\BB$, which is a weak simulation $\zero \to e$ in $\cat{Emp}_\BB^{\leq}$.
\end{proof}

Thus $e:S$ is contextual if and only if the map $\EMP(\Zero)=\enset{\zero}\to\EMP(S)$ corresponding to it arises from a procedure $\Zero\to S$. This suggests the following, more general question: which functions $\EMP(S)\to \EMP(T)$ are induced by a procedure $S\to T$? That is, which such functions are equal to $\EMP(f)$ for some $\fdec{f}{S}{T}$?

Such generalisation can be seen as an instance of Grothendieck's relative point of view \cite{nlab:relativepov}, 
which, roughly speaking, suggests that properties of mathematical objects should be seen as arising from properties of morphisms, so that a geometric object $X$ is defined to have a property precisely when the canonical morphism $X\to 1$ has the corresponding property of morphisms. As, heuristically speaking, one expects geometry to be dual to algebra, we find ourselves in a dual situation: a model can be defined to be contextual if the map $\EMP(\Zero)=\enset{\zero}\to\EMP(S)$ corresponding to it is contextual. To make this precise, we need to answer the question of when a general map $\EMP(S)\to \EMP(T)$ is contextual. 

\subsection{Non-local games as experiments}

Another familiar concept from the literature that admits a neat formulation in our framework is that of non-local games,
or more generally, contextual games as considered in \cite[Appendix E]{abramsky2017contextual}.
Note that these can also be thought of as linear inequalities on the probabilities predicted by empirical models, and thus encompass in particular Bell locality inequalities or non-contextuality inequalities.

One usually thinks of non-local games as follows: there are $n$ spatially  separated parties who can agree on a strategy beforehand but cannot communicate once the game is in play.
A referee sends to the $i$-th party a question from an input set $X_i$,
and expects back an answer from some output set $O_i$.
The referee draws the questions according to a joint probability distribution on inputs.
Afterwards, the referee collects the answers and applies a rule
\[\fdec{W}{X_1\times\cdots\times X_n\times O_1\times\cdots\times O_n}{\enset{0,1}}\]
to determine whether the players win or lose.
Both the distribution on inputs and this winning condition are known a priori by all players.
The goal of the players is to maximise their winning probability.

Such games turn out to be readily formalised as probabilistic experiments: one first builds a (Bell-type) measurement scenario $S$ whose the measurements are given by $\coprod_i X_i$ and where the maximal measurement contexts are given by a choice of a measurement $x_i\in X_i$ for each party.
The outcome set of each measurement in $X_i$ is $O_i$.
If the referee asks a joint question $q \in X_1 \times \cdots \times X_n$ with probability $r_q$, then we can model the game as a probabilistic experiment on $S$, which with probability $r_q$ chooses the measurement  context corresponding to $q$, and then uses $W$ to obtain an outcome in $\enset{0,1}$, which is the outcome set of the only measurement of the scenario $\two$.
Thus, the game setup (namely, the distribution on inputs and the winning condition) can be packaged into a probabilistic experiment $\fdec{g}{S}{\two}$.
From this point of view, a strategy for the players to answer the questions is given by an empirical model $e$ on the scenario $S$.
Finally, the \emph{winning probability} achieved by that strategy arises as the model $\EMP(g)\, e \colon \two$,
bearing in mind that an empirical model on $\two$ is uniquely determined by the probability of obtaining the outcome $1$.

Note that we can also encode a more general class of games, where the referee attributes a pay-off to each combination of questions and answers, rather than a discrete win-or-lose valuation. Such is in particular the general form of a Bell inequality. If we normalise such pay-offs, this amounts to generalising $W$ to take values in $[0,1]$ instead of $\enset{0,1}$. These games can still be seen as probabilistic experiments $\fdec{g}{S}{\two}$, given that the extra flexibility can be modelled through a further mixture of deterministic experiments.

\begin{example}
Recall the Four Candles scenario $\square$ from Example~\ref{ex:fourcandles}. We will represent the CHSH game as a probabilistic experiment $\fdec{\sum_{1=1}^n \sfrac{1}{4}f_i}{\square}{\two}$ with each $f_i$ a deterministic experiment. Recall that a deterministic experiment on $\square$ is determined by specifying a measurement context $\sigma\in \Sigma_\square$ and a subset of $\Ev_\square(\sigma)$ that gets sent to $1$, which we should think of as the winning condition. With this in mind, we let $f_1$ be determined by the context $\enset{\textit{`SammyA'},\, \textit{`EvilG'}}$ and the winning subset be given by $\enset{(\textit{`grape'},\textit{`grain'}),(\textit{`grain'},\textit{`grape'})}$. In other words, for the players to win, $\textit{`SammyA'}$ and $\textit{`EvilG'}$ must anti-coordinate in their choice of beverage. The other experiments $f_2,f_3$ and $f_4$ correspond to the remaining maximal contexts, with the winning constraint given now by coordination, \ie by the set $\enset{(\textit{`grape'},\textit{`grape'}),(\textit{`grain'},\textit{`grain'})}$. 

If all our agents have stubbornly, \ie deterministically, chosen what their chosen beverage is, they can satisfy the referee in at most three out of the four possible questions. As $\EMP(f)$ preserves convex combinations, no classical (\ie non-contextual) strategy can do any better than this. So the classical value of the game, \ie the maximum winning probability achievable by a classical strategy, is $\sfrac{3}{4}$. This bound can be saturated by \eg everyone going for $\textit{`grain'}$.  

However, if the players can coordinate using quantum resources, they can implement the model from Example~\ref{ex:ns}.
This results in a strategy that wins with probability $\sfrac{13}{16} \approx 81\%$.
In fact, with a more judicious choice of shared state and measurements one can design a quantum strategy that wins with probability $\sfrac{(2+\sqrt{2})}{4}$, or approximately $85\%$ of the time.
If `super-quantum' models are allowed, the players can follow a \stress{perfect strategy}, one which allows them to win the game with certainty, namely by using the empirical model from Example~\ref{ex:pr}. This super-quantum model, usually known as the PR box, can in fact be characterised as the unique empirical model that yields a perfect strategy for this game.
\end{example}

\subsection{Predicates and possibilistic models}

When one only cares about perfect strategies, \ie winning a game with certainty,
the only relevant information is contained in the supports of the game and the model, \ie their possibilistic collapses.
As it happens, it will be useful later to restrict attention to only those models that win a particular game with certainty, thus interpreting a (possibilistic) game as a predicate on empirical models. We now make some remarks on this situation.

\begin{definition}\label{def:predicatesatisfaction}
A possibilistic model $e\colon S$ is said to satisfy a possibilistic predicate  $\fdec{g}{S}{\two}$ if
$\EMP_\BB(g)\, e$ is equal to the empirical model on $\two$ corresponding to the deterministic outcome $1$.
A probabilistic model $e\colon S$ is said to satisfy a possibilistic predicate $\fdec{g}{S}{\two}$ if its possibilistic collapse does so.

We write $e\colon \tuple{S,g}$ to indicate that $e\colon S$ satisfies $g$.
The set of all probabilistic (resp. possibilistic) models on $S$ that satisfy $g$ is denoted by $\EMP(S,g)$ (resp. $\EMP_\BB(S,g)$).

There is a preorder on possibilistic predicates $S\to \two$ according to which $g\leq h$ whenever $\EMP_\BB(S,g)\subseteq\EMP_\BB(S,h)$. This yields an equivalence relation on predicates: $g$ and $h$ are said to be equivalent, written $g \sim h$, if $g \leq h$ and $h\leq g$, \ie if $\EMP_\BB(S,g)=\EMP_\BB(S,h)$.
\end{definition}

\begin{definition}\label{def:predicatefrommodel}
A possibilistic model $e\colon S$ induces a predicate $\fdec{g(e)}{S}{\two}$ given by $g(e) \defeq \bigvee_{\sigma\in\Sigma} g(e)_\sigma$ where $g(e)_\sigma$ is the deterministic procedure consisting of the relation mapping the unique measurement $*$ to the context $\sigma$ and the characteristic function $\Ev(\sigma)\to\enset{0,1}=O_*$ of $e_\sigma\subseteq\Ev(\sigma)$.
\end{definition}

In other words, the predicate $g(e)$ is satisfied by a model $d \colon S$ if for each context $\sigma$ the support of $d_\sigma$ is contained in that of $e_\sigma$, \ie if $d \leq e$ as possibilistic models (note that we may need to take the possibilistic collapse of $d$).

\begin{example}\label{ex:ksmodels}
Kochen--Specker models, studied in~\cite[Section 7]{ab} and in~\cite{ruishizzle:extendability}, are certain possibilistic empirical models that are inspired by and closely connected to the original formulation of the Kochen--Specker theorem~\cite{ks}.
These models are easy to characterise in the current framework as those satisfying a certain possibilistic predicate. Let $S$ be a scenario in which each measurement is dichotomic, \ie has outcome set $\enset{0,1}$. For each maximal context $\sigma\in\Sigma_S$ of $S$, define a deterministic predicate $g_\sigma$ corresponding to the subset of $\Ev_S(\sigma)$ consisting of
the local assignments over $\sigma$ which assign outcome $1$ to exactly one measurement in $\sigma$.
More formally, $\pi_{g_\sigma}$ maps the unique measurement $* \in X_\two$ to $\sigma$,
while $\alpha_{g_\sigma}$ sends the subset $\setdef{s\in\Ev(\sigma)}{s^{-1}(1)\text{ is a singleton}}$ to $1$ and the rest of $\Ev(\sigma)$ to $0$. Define the Kochen--Specker predicate $g_{\text{KS}}$ on $S$ as $\bigvee_{\sigma}g_\sigma$ where $\sigma$ ranges over all maximal contexts. For a scenario $S$, we then say that $e$ is a Kochen--Specker model if it satisfies the predicate $g_{\text{KS}}$. From this point of view, the import of the Kochen--Specker theorem is that, for suitably chosen scenarios of quantum measurements:
\begin{enumerate}
    \item any quantum state gives rise to a Kochen--Specker model on the scenario;
    \item any Kochen--Specker model on the scenario is strongly contextual.
\end{enumerate}
This results in a \emph{state-independent strong contextuality} argument.\footnotemark\
The results of~\cite[Section 7]{ab} can then be seen as giving a criterion on a scenario $S$ that implies the strong contextuality of every Kochen--Specker model on it.
\end{example}

\footnotetext{See \cite[Section 4]{abramsky2017quantum} for a more general account of state-independent contextuality phrased in similar language.}

Kochen--Specker models are defined somewhat differently in~\cite{ruishizzle:extendability}.
In contrast to the viewpoint above, there is single Kochen--Specker model on each scenario according to the definition there. Namely, this is defined to be the possibilistic empirical model whose support for each context is given exactly by the winning constraint of the Kochen--Specker predicate above.
Importantly, this support of the predicate is guaranteed to satisfy (possibilisitic) no-signalling, and thus it yields a well-defined possibilistic model.
In other words, it turns out that the Kochen--Specker predicate $g_{\text{KS}}$ is induced by a possiblistic model in the sense of Definition~\ref{def:predicatefrommodel}; that is, $g_{\text{KS}}=g(e)$ for some possibilistic model $e$.

This is no accident. In fact, it is true of every possibilistic predicate, at least up to equivalence in the predicate preorder (Definition~\ref{def:predicatesatisfaction}).

 \begin{proposition}\label{prop:possibilisticgamesasmodels}
If a possibilistic predicate $\fdec{g}{S}{\two}$ is satisfied by some model, then it is equivalent to a predicate $g(e)$ induced by some possibilistic model $e \colon S$.
 \end{proposition}
\begin{proof}
As long as $\EMP_\BB(S,g) \neq \emptyset$, the predicate $g$ is equivalent to $g_e$ where $e=\bigvee\setdef{d}{d\in \EMP_\BB(S,g)}$.
\end{proof}

This result means that for the purpose of satisfaction by empirical models, we might as well assume that all predicates are either the predicate \stress{false} (which constantly outputs the outcome $0$) or of the form $g(e)$. Note that not only are all satisfiable predicates equivalent to one of this form, but indeed there is a canonical representative for every equivalence class, given by the choice of $e \colon S$ in the proof above.

In a sense, one can think of an arbitrary predicate as defining a `theory' that we are interested in having satisfied by empirical models. Then, passing to the canonical form corresponds to closing this theory under implication assuming only no-signalling as an axiom.

Note that if one instead works with the definition that $g\leq h$ 
when $\EMP(S,g)\subseteq\EMP(S,h)$, \ie when any \stress{probabilistic} model satisfying $g$ also satisfies $h$, one can prove that each satisfiable predicate $\fdec{g}{S}{\two}$ is equivalent to $g(e)$ where $e$ is the possibilistic collapse of some probabilistic model, rather than an arbitrary possibilistic model. One can just take $e=\bigvee\setdef{d}{d\in \EMP(S,g)}$. 

\section{When is a function induced by a procedure?}\label{sec:main-result}

We now turn our attention towards answering the central question: which functions $\fdec{F}{\EMP(S)}{\EMP(T)}$ are induced by a probabilistic procedure $\fdec{f}{S}{T}$? 
We will build up to to the answer in stages: we start with the case where $T\cong\nn$ and the function $F$ preserves deterministic models,
culminating with an answer in the deterministic case in Theorem~\ref{thm:characterizingdeterministicprocedures} and in the general case in Theorem~\ref{thm:mainthm}.

\subsection{The affine span of deterministic models}

Before diving in, we make some useful preliminary observations.
First, note that Lemma~\ref{lem:simsareconvex} implies that it is necessary for $F$ to preserve convex combinations for it to be induced by a procedure. The next two results enable us to deduce that
a convex-combination-preserving $\fdec{F}{\EMP(S)}{\EMP(T)}$ 
is uniquely determined by its action on deterministic empirical models on $S$,
which we can identify with $\Ev(X_S)$.

A central result of~\cite{ab} is that the no-signalling condition is equivalent to
realisability by a non-contextual hidden variable model with `negative probabilities'.
In other words, any (no-signalling) probabilistic empirical model $e \colon S$ can be extended
to a quasiprobability distribution on global assignments; \ie there is function $\fdec{d}{\Ev(X_S)}{\RR}$ with $\sum_{s \in \Ev(X_S)}d(s) = 1$ which marginalises to yield $d|_\sigma = e_\sigma$ for every context $\sigma \in \Sigma_S$.
We state this result here in the form we will require.

\begin{theorem}[Theorem 5.9 of~\cite{ab}, rephrased]
Any probabilistic empirical model $e \colon S$ can be written as an affine combination of deterministic empirical models, \ie $e = \sum_{s\in\Ev(X_S)} r_s \delta_s$ for some $r_s \in \RR$ with $\sum_s r_s = 1$.
\end{theorem}

\begin{lemma}\label{lem:simsareaffine}
If $\fdec{F}{\EMP(S)}{\EMP(T)}$ preserves convex combinations, then it preserves existing affine combinations. That is, if $e=\sum r_i e_i$ where $e, e_i\colon S$ and $r_i\in\RR$ satisfy $\sum r_i=1$, then $F(e)=\sum r_i F(e_i)$.
\end{lemma}
\begin{proof}
We prove this when $e=r_1e_1+r_2e_2$, the general case following by induction.
Without loss of generality we may assume that $r_1 \geq r_2$. If $r_1\leq 1$, then this is an ordinary convex combination and there is nothing to prove. If $r_1>1$, then
rearranging the equation into $e_1=(\sfrac{1}{r_1})e-(\sfrac{r_2}{r_1})e_2$ results in an ordinary convex combination. Plugging this into $F$, which we know to preserve convex combinations, and then rearranging back gives $F(e)=r_1F(e_1)+r_2F(e_2)$ as desired.
\end{proof}

From the last two results we immediately obtain the following. 

\begin{theorem}\label{thm:convexfnsdeterminedatsections}
A convex-combination preserving function $\fdec{F}{\EMP(S)}{\EMP(T)}$ is uniquely determined by its restriction to $\Ev(X_S)$; \ie if $F,G\colon \EMP(S) \doubleto\EMP(T)$ preserve convex combinations and agree on deterministic models, then $F=G$.
\end{theorem}
 
\subsection{The case of deterministic experiments}

We first consider the crucial special case of functions
$\fdec{F}{\EMP(S)}{\EMP(T)}$ that preserve deterministic models.
As deterministic models can be identified with global assignments, such an $F$ induces a function $\Ev(X_S)\to\Ev(X_T)$. By Theorem~\ref{thm:convexfnsdeterminedatsections}, it is in fact determined by this restriction.

Specialising further to the case when $T=\nn$, for which $\Ev(X_T) = \Ev(\enset{*}) = O_{\nn,*} =\enset{0,\ldots,n-1}$,
it is easy to see when such a function arises from a deterministic experiment $S\to\nn$. This happens precisely when $F$ can be computed within a fixed context.

\begin{proposition}\label{prop:testsaslocaltests}
A function $\fdec{f}{\Ev(X_S)}{\enset{0, \ldots, n-1}}$ arises from a deterministic experiment $S\to\nn$ if and only if it factors through $\Ev_S(\sigma)$ for some $\sigma\in\Sigma_S$.
\end{proposition}

The following lemma will be crucial to our characterisation.
It contains a result that can actually be phrased in
purely classical (\ie non-contextual, even deterministic) language.
It concerns deterministic functions of the form 
\[\fdec{f}{O_1\times \cdots \times O_k}{Q}\]
which we think about as a function on $k$ arguments.
One might be interested in calculating the result of the function by inspecting
as few of the arguments as possible.
However, deciding which of them to inspect must be done statically,
before any of their values are known.
The point of the lemma is that there is always a canonical optimal choice,
a least subset of arguments that can get the job done.
This no longer holds if one is permitted the extra flexibility  
of dynamically choosing the next argument to inspect depending on previously observed values.
Indeed, this constitutes the main obstacle to answering question \eqref{q:proc} for adaptive procedures, which have such flexibility~\cite{ourlics:comonadicview}.

\begin{lemma}\label{lem:leastsubset}
Let $S$ be a scenario and  $\fdec{F}{\Ev(X_S)}{Y}$ a function. Then there is a least subset $U\subseteq X_S$ such that $F$ factors through $\Ev(U)$. 

Moreover, let $\family{F_i}_{i \in I}$ be a family of functions $\fdec{F_i}{\Ev(X_S)}{Y_i}$,
and let $U_i$ be the least subset such that $F_i$ factors through $\Ev(U_i)$. Then the least subset $U$ such that $\fdec{\tuple{F_i}_{i\in I}}{\Ev(X)}{\prod_{i\in I} Y_i}$ factors through $\Ev(U)$ equals $\bigcup_{i \in I} U_i$.
\end{lemma}
\begin{proof}
For the first part it suffices to show that if $F$ factors through $\Ev(V)$ and through $\Ev(W)$, then it factors through $\Ev(V\cap W)$. Note first that $F$ factors through $\Ev(U)$ if and only if for all $s,t \in \Ev(X_S)$, $s|_U=t|_U$ implies $F(s)=F(t)$. Now, assume that $F$ factors through $\Ev(V)$ and through $\Ev(W)$, and let $s, t \in \Ev(X_S)$ satisfying $s|_{V\cap W}=t|_{V\cap W}$. Define $r$ by gluing together $s|_V$ and $t|_{X\setminus V}$. This satisfies $r|_V=s|_V$ and $r|_W=t|_W$ by construction. As $F$ factors through $\Ev(V)$ and through $\Ev(W)$, we conclude that $F(s)=F(r)=F(t)$, as desired.  

For the second part, let $F_i$, $U_i$ and $U$ be as in the statement. Since $\bigcup_i U_i$ contains each $U_i$, the function $F_i$ factors through $\Ev(\bigcup_i U_i)$ and thus so does $\tuple{F_i}_{i \in I}$.
Hence, $U\subseteq\bigcup_i U_i$. On the other hand, each $F_i$ factors through $\Ev(U)$, whence $U$ contains each $U_i$ and thus it contains $\bigcup_i U_i$, completing the proof.
\end{proof}

We discuss a more abstract interpretation of this result in Section~\ref{ssec:variantsofmainresult},
when we examine why its analogue fails in the adaptive case.

\subsection{The case of deterministic procedures}\label{ssec:casedetmaps}

We will now leverage Lemma~\ref{lem:leastsubset} to answer when a function $\fdec{F}{\EMP(S)}{\EMP(T)}$  is induced by a deterministic procedure $S\to T$. As deterministic procedures map deterministic models to deterministic models, any $\fdec{F}{\EMP(S)}{\EMP(T)}$ induced by a deterministic procedure must do so as well. For such a function, the restriction to deterministic models induces a function $\fdec{\hat{F}}{\Ev(X_S)}{\Ev(X_S)}$ fitting into a commutative square 
\[\begin{tikzpicture} 
         \matrix (m) [matrix of math nodes,row sep=2em,column sep=5em,minimum width=2em]
         {
          \Ev_S(X_S) & \Ev_T(X_T)  \\
          \EMP(S) & \EMP(T)  \\};
         \path[->]
         (m-1-1) edge node [left] {} (m-2-1)
                edge node [above] {$\hat{F}$} (m-1-2)
         (m-1-2) edge node [right] {} (m-2-2)
         (m-2-1) edge node [below] {$F$} (m-2-2);
  \end{tikzpicture}\]
Moreover, we know that $F$ must preserve convex combinations, so that  $F$ is determined uniquely by $\hat{F}$ by Theorem~\ref{thm:convexfnsdeterminedatsections}. We can now use $\hat{F}$ to characterize when $F$ is induced by a deterministic procedure. 

\begin{theorem}\label{thm:characterizingdeterministicprocedures}
    Let  $\fdec{F}{\EMP(S)}{\EMP(T)}$ be a function preserving convex combinations and deterministic models. Then $F$ is induced by a deterministic procedure $S\to T$ iff for each context $\sigma\in \Sigma_T$ the composite \[\Ev_S(X_S)\xrightarrow{\hat{F}}\Ev_T(X_T)\to \Ev_T(\sigma)\] factors through $\Ev_S(\tau)$ for some $\tau\in\Sigma_S$.
\end{theorem}

\begin{proof}
    Assume first that $F=\EMP(f)$, and let $\sigma\in \Sigma_T$. Then the composite $\Ev_S(X_S)\xrightarrow{\hat{F}}\Ev_T(X_T)\to \Ev_T(\sigma)$ factors trough $\Ev_S(\pi_f(\sigma))$ and $\pi_F(\sigma)\in\Sigma_S$ as $\pi$ is simplicial. 

    For the converse, let us define a procedure $f\colon S\to T$ inducing $F$ as follows. Given $x\in X_T$, let $\pi_f(x)$ be the least subset of $X_S$ such that $\Ev_S(X_S)\xrightarrow{\hat{F}}\Ev_T(X_T)\to \Ev_T(x)$ factors through $\Ev_S(\pi_f(x))$, as guaranteed by Lemma~\ref{lem:leastsubset}. As restriction maps are surjective, the map $\Ev_S(\pi_f(x))\to \Ev_T(x)$ appearing in the factorization of $\Ev_S(X_S)\xrightarrow{\hat{F}}\Ev_T(X_T)\to \Ev_T(x)$
    is uniquely determined, and we define $\alpha_{f,x}$ to be that map. 

    This defines $f=(\pi_f,\alpha_f)$. We now show that $\pi_f$ is a simplicial relation. Indeed, for any $\sigma\in\Sigma_S$, the second part of Lemma~\ref{lem:leastsubset} implies that $\pi_f(\sigma)=\bigcup_{x\in \sigma} \pi_f(x)$ is the least subset such that $\Ev_S(X_S)\xrightarrow{\hat{F}}\Ev_T(X_T)\to \Ev_T(\sigma)$ factors through $\Ev_S(\pi_f(\sigma))$. By our assumption on $\hat{F}$, we must have $\pi_f(\sigma)\in\Sigma_S$, so that $\pi_f$ is a simplicial relation.

    Finally, $f$ is constructed so that the induced map $\Ev_S(X_S)\to\Ev_S(\pi_f(x))\to \Ev_T(x)$ is equal to $\Ev_S(X_S)\xrightarrow{\hat{F}}\Ev_T(X_T)\to \Ev_T(x)$ for each $x\in X_T$. This implies that both $F$ and $\EMP(f)$ agree when restricted to deterministic models, whence $F=\EMP(f)$ by Theorem~\ref{thm:convexfnsdeterminedatsections}.
\end{proof}

\subsection{The case of probabilistic experiments}\label{ssec:caseprobabilisticexperiments}

We now generalise Proposition~\ref{prop:testsaslocaltests} in order to answer the question of when a function $F\colon \EMP(S)\to\EMP(\nn)$ is induced by a probabilistic experiment $S\to\nn$. Assuming that $F$ preserves convex combinations, Theorem~\ref{thm:convexfnsdeterminedatsections} implies that we can simplify this situation by restricting $F$ to $\Ev(X_S)$.

Note that $\EMP(\nn) \cong D(O_{\nn,*}) \cong D(\enset{0,\ldots n-1})$.
So we aim to characterise the functions 
\[\fdec{F}{\Ev(X_S)}{D(\enset{0,\ldots n-1})}\]
that are induced by probabilistic experiments $S\to\nn$. 
This happens if and only if 
$F$ can be written as a convex combination
$F = \sum_j r_jF_j$, with $r_j \geq 0$ and $\sum_j r_j = 1$,  where each \[\fdec{F_j}{\Ev(X_S)}{\enset{0,\ldots, n-1}}\] is a deterministic function
factoring through a projection $\Ev(X_S)\to\Ev(\sigma_j)$ for some context $\sigma_j\in \Sigma_S$.
This follows from the definition of probabilistic procedures as convex mixtures of deterministic ones, the fact that $\EMP$ preserves convex mixtures, \ie
\[\EMP(\sum_j r_jf_j) = \sum_j r_j \EMP(f_j),\]
and the characterisation of deterministic experiments given in Proposition~\ref{prop:testsaslocaltests}.

\subsection{The general case and procedures as empirical models}

We now wish to answer the general question of when a function $F\colon \EMP(S)\to\EMP(T)$ is induced by a probabilistic procedure $S\to T$.
In order to understand $F$ we can study the composites
\[\EMP(S)\to\EMP(T)\to\EMP(T|_\sigma)\] for each context $\sigma\in\Sigma_T$.
Note that each of these is an instance of the simpler setting considered in Section~\ref{ssec:caseprobabilisticexperiments}.
As before, we may assume that $F$ preserves convex combinations, so that it is determined by its values on $\Ev(X_S)$.
Since in addition $\EMP(T|_\sigma) \cong D(\Ev_T(\sigma))$,
each composite above is a function of type
\[\fdec{F_\sigma}{\Ev_S(X_S)}{D(\Ev_T(\sigma))} \Mdot\]
When the map $F$ is induced by a probabilistic procedure,
then each $F_\sigma$ is described as a convex mixture of (deterministic) functions
\[\fdec{F_{\sigma,i}}{\Ev_S(X_S)}{\Ev_T(\sigma)} \Mcomma\]
each of which must factor through $\Ev(\tau)$ for some $\tau$ a context of $S$ by Proposition~\ref{prop:testsaslocaltests}.
So, for each context $\sigma\in \Sigma_T$ we get some probabilistic data, namely a probability distribution on such functions. 
It also turns out that this data glues well across different $\sigma \in \Sigma_T$.
This suggests thinking of it as an empirical model on a new scenario.
We now make this idea precise.

\begin{definition}
Let $S$ and $T$ be measurement scenarios. We define a new scenario $[S,T]$
by setting $X_{[S,T]} \defeq X_T$, $\Sigma_{[S,T]}\defeq \Sigma_T$, and 
\[O_{[S,T],x}\defeq\setdef{\tuple{U,\alpha}}{U\subseteq X_S, \fdec{\alpha}{\Ev_S(U)}{O_{T,x}}} \Mdot\]
We equip $[S,T]$ with the possibilistic predicate $g_{S,T}$ defined by $g_{S,T}\defeq\bigvee_{\sigma \in \Sigma_T}g_\sigma$ where $g_\sigma$ checks that only a compatible part of $S$ is used in the simulation of the context $\sigma\in\Sigma_T$. More formally, $g_\sigma$ corresponds to the subset \[\setdef{\family{\tuple{U_x,\alpha_x}}_{x\in \sigma}}{\bigcup_{x \in \sigma} U_x\in\Sigma_S} \; \subseteq \; \Ev_{[S,T]}(\sigma) \Mdot\]
\end{definition}

\begin{proposition}\label{prop:modelson[S,T]}
Deterministic procedures $S \to T$ correspond bijectively to
deterministic empirical models of $[S,T]$ satisfying $g_{S,T}$.
\end{proposition}

\begin{proof}
Deterministic models are determined by global assignments. Such an assignment
$s\in\Ev(X_{[S,T]})$ consists of an outcome for each measurement $x\in X_{[S,T]}=X_T$.
Each such outcome is a pair $\tuple{U_x,\alpha_x}$ consisting of a subset of measurements of $S$, $U_x \subseteq X_S$, and a function $\fdec{\alpha_x}{\Ev_S(U_x)}{\Ev_T(\enset{x})}$.

This is almost exactly the data needed to specify a deterministic procedure $\fdec{f = \tuple{\pi_f,\alpha_f}}{S}{T}$, with $\pi_f$ defined so that $\pi_f(x)\defeq U_x$. The only caveat is that the relation $\fdec{\pi_f}{X_T}{X_S}$ thus defined need not be simplicial (with respect to the complexes $\Sigma_T$ and $\Sigma_S$).
It turns out that it is a simplicial relation if and only if the deterministic model $\delta_s$ satisfies the predicate  $g_{S,T}$. 
\end{proof}

\begin{corollary}\label{cor:modelson[S,T]}
Probabilistic procedures $S\to T$ correspond bijectively to probability distributions on
\[\setdef{s \in \Ev_{[S,T]}(X_{[S,T]})}{\text{$\delta_s$ satisfies $g_{S,T}$}} \Mcomma\]
and thus give rise to all non-contextual models $e\colon\tuple{[S,T],g_{S,T}}$.
\end{corollary}

\begin{proof}
Recall that a probabilistic procedure is a convex combination of deterministic procedures.
The statement thus follows from applying $\Dist$ to both sides of the bijection from
Proposition~\ref{prop:modelson[S,T]}
between deterministic procedures $S \to T$ and deterministic models in $\tuple{[S,T],g_{S,T}}$.

Explicitly, given a probabilistic procedure $\sum r_i f_i$ with each $\fdec{f_i}{S}{T}$ a deterministic procedure,
and writing $s_i \in \Ev(X_{[S,T]})$ for the global assignment corresponding to $f_i$, then
$\sum r_i \delta_{s_i} \colon \tuple{[S,T],g_{S,T}}$ is the empirical model corresponding to $\sum r_i f_i$. It satisfies $g_{S,T}$ because each $\delta_{s_i}$ does, and moreover, it is non-contextual by construction.

Conversely, any non-contextual model in $\tuple{[S,T],g_{S,T}}$ can be written (not necessarily uniquely) as a convex combination of deterministic models $\delta_s$ satisfying $g_{S,T}$, so it arises in this fashion. 
\end{proof}

We have seen in Definition~\ref{def:EMPprobabilistic} that probabilistic procedures $\fdec{f}{S}{T}$ determine convex-preserving maps $\fdec{\EMP(f)}{\EMP(S)}{\EMP(T)}$ between sets of empirical models.
If two procedures induce the same (non-contextual) model on $[S,T]$ as in Corollary~\ref{cor:modelson[S,T]}, then both determine the same action on empirical models.
This is because $\EMP(f)$ depends only on how $f$ behaves at each context of $T$, which is precisely the information retained in the empirical model induced by $f$.
One may then wonder, what about other, possibly contextual, empirical models?

In other words, the situation is represented in the diagram below, which we now explain.
The map sending a probabilistic procedure $f$ to its action $\EMP(f)$ on empirical models factors through the map sending $f$ to the noncontextual model $e_f \colon \tuple{[S,T],g_{S,T}}$ induced by it.
Hence, there is a map (depicted by a dashed arrow in the diagram) from the set $\textsf{NC}([S,T],g_{S,T})$ of noncontextual empirical models on $[S,T]$ satisfying $g_{S,T}$ to the set of maps $\EMP(S)\to\EMP(T)$.
The question is then whether this map can be extended to arbitrary empirical models in a reasonable way,
as depicted by the dotted arrow labelled by a question mark in the diagram.

 \[\begin{tikzpicture} 
         \matrix (m) [matrix of math nodes,row sep=5em,column sep=5em,minimum width=2em]
         {
          \Scen(S,T)  & \cat{Set}(\EMP(S),\EMP(T))  \\
         \textsf{NC}([S,T],g_{S,T}) & \EMP([S,T],g_{S,T})  \\};

         \path[->] (m-1-1) edge node [left] {$f\mapsto e_f$} (m-2-1);
         \path[->] (m-1-1) edge node [above] {$f\mapsto\EMP(f)$} (m-1-2);
         \path[right hook->](m-2-1) edge (m-2-2);
         \path[->,dashed] (m-2-1) edge (m-1-2);
         \path[->,dotted] (m-2-2) edge node[right] {$e\mapsto F_e$~{\large?}} (m-1-2);
  \end{tikzpicture}\]

Indeed, any empirical model $e \colon \tuple{[S,T},g_{[S,T]})$ gives rise to a convex-preserving map $F_e\colon \EMP(S) \to \EMP(T)$.
Intuitively, given a context $\sigma$ of measurements in $T$, sampling from $e_\sigma$ tells us what context to measure in the scenario $S$ and how to interpret the outcomes as outcomes for $\sigma$. In this way, any empirical model on $S$ gets mapped to one on $T$.

Let us now make this more precise.
We begin by analysing the structure of an arbitrary empirical model $e \colon \tuple{[S,T],g_{[S,T]}}$. 
For each $\sigma\in \Sigma_T$, the support of $e_\sigma$ consists of tuples of the form $\family{\tuple{U_x,\alpha_x:\Ev_S(U_x)\to O_{T,x}}}_{x \in \sigma}$,
where the fact that $e$ satisfies $g_{S,T}$ implies that
$\bigcup_{x \in \sigma} U_x \in \Sigma_S$.
As a result, we can equivalently view $e_\sigma$  as a probabilistic procedure  $f_{e,\sigma}\colon S\to T|_\sigma$.
Moreover, the no-signalling condition for $e$ means that these probabilistic procedures fit together in the sense that given $\tau\subseteq \sigma$ in $\Sigma_T$, the procedure $f_{e,\tau}$ can be obtained from $f_{e,\sigma}$ by restricting to $\tau$. In other words, the triangle
 \[\begin{tikzpicture} 
         \matrix (m) [matrix of math nodes,row sep=2em,column sep=5em,minimum width=2em]
         {
          S & T|_\sigma  \\
           & T|_\tau  \\};
         \path[->]
         (m-1-1) edge node [below] {$f_{e,\tau}$}  (m-2-2)
                edge node [above] {$f_{e,\sigma}$} (m-1-2)
         (m-1-2) edge node [right] {} (m-2-2);
  \end{tikzpicture}\]
commutes.
Note that such a family $\family{\fdec{f_{e,\sigma}}{S}{T|_\sigma}}_{\sigma\in\Sigma_T}$ of compatible  probabilistic procedures
does not necessarily arise as the restrictions of a single (global) procedure $S \to T$.
Such glueing is possible precisely when $e$ is noncontextual.

We now use this perspective on an empirical model $e \colon \tuple{[S,T],g_{[S,T]}}$ to define the map $\fdec{F_e}{\EMP(S)}{\EMP(T)}$. We do this by setting $F_e(d)_\sigma\defeq\EMP(f_{e,\sigma})(d)$: the no-signalling condition on $e$ then implies that $F_e(d)$ is a no-signalling model on $T$ for any $d \colon S$.

The discussion above establishes the following theorem.

\begin{theorem}\label{thm:mainthm}
A function $F\colon \EMP(S)\to\EMP(T)$ preserving convex combinations is induced by a probabilistic procedure $S\to T$ if and only if $F = F_e$ for some $e\colon [S,T]$ that is non-contextual and satisfies $g_{S,T}$.
\end{theorem}

Observe that
taking $S$ to be the trivial scenario $\Zero$
this result reduces to (the probabilistic case of) Theorem~\ref{thm:contextualiffnonsimulable},
which characterises noncontextual models in terms of simulability from the trivial model.
More generally, the question of characterising whether a given map $\fdec{F}{\EMP(S)}{\EMP(T)}$ is induced by a procedure
turns into this special case of testing for noncontextuality.
More precisely, it is not quite testing for membership in the convex polytope  of noncontextual empirical models $\textsf{NC}([S,T])$, or even $\textsf{NC}([S,T],g_{S,T})$,
but in the image of that polytope under the quotient map from $\EMP([S,T],g_{S,T})$ 
to the space of convex-preserving transformations $\EMP(S) \to \EMP(T)$.
Since this map is linear, the membership question is still, just like contextuality testing, a linear programming problem.
Actually, testing membership in a projection of the noncontextual polytope often crops up in study of nonlocality and contextuality, too.
For example, certain Bell or noncontextuality inequalities -- notably, the famous CHSH inequality --
involve only the values of two-measurement correlators, dispensing with full information about the probability distribution on joint outcomes. The same could be said about other parity-based contextuality proofs such as the GHZ--Mermin model \cite{mermin1990simple,greenberger1989going} or All-versus-Nothing arguments more generally \cite{mermin1990extreme,contextualitycohomologyparadox,abramsky2017complete}. This method of projecting to another polytope is made more explicit in~\cite[Section VII]{abramskyhardy:logical}. As the examples show, such limited information may be sufficient to witness contextuality.

\section{Getting closure}\label{sec:closure}

The results of the previous section strongly suggest thinking of $[S,T]$ as something like an \stress{internal hom}.
However, the construction of $[S,T]$ does not take into account the measurement compatibility structure of $S$ encoded in the simplicial complex $\Sigma_S$ -- that is taken care of by the predicate $g_{S,T}$.
This suggests that one should work with pairs $\tuple{S,\fdec{g}{S}{\two}}$ as the basic objects of our category.

In this section we make this viewpoint precise and show that it results in a \emph{closed category}, a notion introduced in ~\cite{eilenbergkelly:closedcats}. Instead of the original definition, we work with the axiomatisation given in~\cite{manzyuk:closedcats} and attributed to~\cite{laplaza:closedcats} and~\cite[Section 4]{Street:cosmoi}. Roughly speaking, closed categories are like monoidal closed categories without the monoidal structure:
they axiomatise the notion of a category \cat{C} where the collection of morphisms $A\to B$ can be given the structure of a \cat{C}-object. This amounts to defining a bifunctor $[-,-]\colon \cat{C}\op\times\cat{C}\to\cat{C}$ with suitable structure on it satisfying some coherence conditions. The idea is that this structure captures the notions of identity and composition. In the absence of a monoidal product, in contrast to the definition of a monoidal closed category, this internal hom structure must be encoded in a \stress{curried} version. For example, composition is represented by a 
transformation $\fdec{L^A_{B,C}}{[B,C]}{[[A,B],[A,C]]}$ natural in $B$ and $C$ and dinatural in $A$.

As the full details are rather technical, we do not recall the complete definition here, even if we do check the conditions in some detail in the proof of Theorem~\ref{thm:closure}.
The reader interested in understanding finer points of the proof should be able to do so after consulting the definition in~\cite{manzyuk:closedcats}.

\begin{definition}\label{def:scenswithgames}
We define a new category whose objects are pairs  $\tuple{S,\fdec{g}{S}{\two}}$ consisting of a scenario and a \emph{structure predicate} on it. We assume that the structure predicate $g$ on $S$
is induced by some (possibilistic) model on $S$ as per Definition~\ref{def:predicatefrommodel}.
A morphism $\tuple{S,g}\to\tuple{T,h}$ is given by a probabilistic procedure $f\colon S\to T$ such that for any probabilistic model $e\colon \tuple{S,g}$ we have $\EMP(f)\,e\colon\tuple{T,h}$, \ie if $e$ satisfies the structure predicate on $S$, then $\EMP(f)\,e$ satisfies the structure predicate on $T$. We denote this category by $\Sceng$. The subcategory of $\Sceng$ with deterministic maps is denoted by $\ScengDet$, while $\Sceng_\BB$ denotes the category obtained by replacing probabilistic procedures and models with possibilistic ones in the preceding definition. 

Given two objects $\tuple{S,g}$ and $\tuple{T,h}$ of $\Sceng$
we define a new object $[\tuple{S,g},\tuple{T,h}]$ by setting
 \[[\tuple{S,g},\tuple{T,h}] \defeq \tuple{[S,T],\bigvee_{\sigma\in\Sigma_T}g_{\sigma,S,T,g,h}}\]
 where $g_{\sigma,S,T,g,h}$ is a deterministic predicate corresponding to the subset of $\Ev_{[S,T]}(\sigma)$ defined by 
    \[\setdef{\family{\tuple{U_x,\fdec{f_x}{\Ev_S(U_x)}{O_{T,x}}}}_{x \in \sigma}}{U \defeq \bigcup_{x \in \sigma}U_x\in\Sigma_S,\Forall{s\in \Ev_S(U)} s\in g_U \Rightarrow f(s) \defeq \family{f_x(s|_{U_x})}_{x\in\sigma} \in h_\sigma} \Mdot\] 
In other words, the predicate on $[\tuple{S,g},\tuple{T,h]}$ checks, for each context $\sigma$ of $\Sigma_T$,  two things: 
    \begin{enumerate}
        \item that only a context of $\Sigma_S$ is used;
        \item that any local assignment satisfying $g$ is mapped to an assignment satisfying $h$.
    \end{enumerate}
\end{definition}

The requirement that the structure predicates be induced by possibilistic models
is so that one can check locally whether $f\colon S\to T$ is in fact a map $\tuple{S,g}\to\tuple{T,h}$. This enables one to write down the definition of the structure predicate on $[\tuple{S,g},\tuple{T,h}]$.
Of course, one should check that this structure predicate is also induced by a possibilistic model. By construction, it is given by a Boolean distribution for each $\sigma$, so one need only check that this is non-signalling.
That is, we need to show that if $\tau \subseteq \sigma$, then any possible $\family{\tuple{U_x,\fdec{f_x}{\Ev_S(U_x)}{O_{T,x}}}}_{x \in \tau}$ can be extended to the context $\sigma$.
Let $U \defeq \bigcup_{x \in \tau} U_x$.
Given an assignment $s \in \Ev_S(U)$, this is already mapped to an assignment $f(s) \in \Ev_T(\tau)$ given by
$f(s) \defeq \family{f_x(s|_{U_x})}_{x\in\tau}$. Moreover, if $s$ satisfies $g_U$ then $f(s)$ satisfies $h_\tau$, and $f(s)$ can therefore be extended to some assignment $t_s \in \Ev_T(\sigma)$ satisfying $h_\sigma$ because $h$ is non-signalling. For $y \in \sigma \setminus \tau$, we set $U_y \defeq U$ and $\fdec{f_y}{\Ev_S(U)}{O_{T,y}}$
given by $f_y(s) \defeq (t_s)_y$ whenever $s \in g_U$ (note that $f_y(s)$ can be set to anything for $s \notin g_U$).

Note that the restriction to predicates induced by possibilistic models is made out of convenience and does not represent a consequential choice. This is because Proposition~\ref{prop:possibilisticgamesasmodels} shows that any predicate is equivalent to one induced by a possibilistic model, in the sense that they have the same set of satisfying models, which is their relevant characteristic here. Therefore, we could have chosen as objects arbitrary pairs of scenarios and predicates and then use the canonical representatives given by Proposition~\ref{prop:possibilisticgamesasmodels} in the definition of $g_{\sigma,S,T,g,h}$.

Note that a plain scenario $S$ can be seen as pair $\tuple{S,\mathsf{t}}$ where $\mathsf{t}$ is a trivial predicate that is always satisfied.
Our earlier categories of scenarios are thus full subcategories of the corresponding categories where objects are scenarios equipped with structure predicates.

\begin{theorem}\label{thm:closure}
 $[-,-]$ makes $\ScengDet$ into a closed category.
\end{theorem}
\begin{proof}

We first show how to make $[-,-]$ functorial in both variables. Given a deterministic morphism $\fdec{\tuple{\pi,\alpha}}{\tuple{S,e}}{\tuple{T,f}}$,
we define a morphism \[\fdec{[\id,\tuple{\pi,\alpha}]}{[\tuple{P,g},\tuple{S,e}]}{[\tuple{P,g},\tuple{T,f}]}\]
for any $\tuple{P,g}$. This boils down to defining a morphism $\fdec{[\id,\tuple{\pi,\alpha}]}{[P,S]}{[P,T]}$ and explaining why it agrees with the relevant predicates. As $\Sigma_S=\Sigma_{[P,S]}$ and similarly for $T$, we can define the simplicial relation underlying $[\id,\tuple{\pi,\alpha}]$ to be $\pi$.
Now, given $x\in X_T$, we define the function $\Ev_{[P,S]}(\pi(x))\to O_{[P,T],x}$ to be ``post-composition with $\alpha_x$''. More precisely, we send a joint assignment in $\Ev_{[P,S]}(\pi (x))$,
\ie a family $\family{\tuple{U_y,k_y}}_{y \in \pi(x)}$ of functions $\fdec{k_y}{\Ev_P(U_y)}{O_{S,y}}$,
to the element of $O_{[P,T],x}$ given by
\[\tuple{\;\;U \defeq \bigcup_{y \in \pi(x)} U_y\;,\;\;\; \alpha_x \circ \tuple{k_y \circ \rho_y}_{y \in \pi(x)}\;\;}\]
where $\fdec{\rho_y \defeq \Ev_P(U_y \subseteq U)}{\Ev_P(U)}{\Ev_P(U_y)}$ is the obvious projection;
\ie we take the composite at the bottom of the following diagram:
\[\begin{tikzcd}
& \Ev_P(U_y)  \arrow[r,above,"k_y"] & O_{S,y} \\
\Ev_P(U) \arrow[ru,"\rho_y"]  \arrow[rr,dashed,"\tuple{k_y\circ \rho_y}_{y \in \pi(x)}"] &  & \Ev_S(\pi x) =\prod\limits_{y\in \pi x}O_{S,y} \arrow[u]
\arrow[rr,"\alpha_x"] & &  \Ev_T(x) \\
\end{tikzcd} \Mdot\]

To see that this results in a morphism  $[\tuple{P,g},\tuple{S,e}]\to [\tuple{P,g},\tuple{T,f}]$, note that a function $k\colon \Ev_P(U)\to\Ev_S(\sigma)$ is an allowed outcome of $[\tuple{P,g},\tuple{S,e}]$ if and only if $U\in\Sigma_P$ and $k$ sends outcomes in the support of $g$ to outcomes in the support of $e$. Whenever this happens, the fact that $\tuple{\pi,\alpha}$ is a morphism  $\tuple{S,e}\to \tuple{T,f}$ implies that postcomposing with $\alpha$ results in a function that sends outcomes in the support of $g$ to outcomes in the support of $f$, and thus $[\id,\tuple{\pi,\alpha}]$ is indeed a morphism  $[\tuple{P,g},\tuple{S,e}]\to [\tuple{P,g},\tuple{T,f}]$. It is clear that this action on morphisms is functorial in the second variable.

Moving now to the first variable, given a deterministic morphism $\tuple{\pi,\alpha}\colon \tuple{S,e}\to \tuple{T,f}$, we wish to define a morphism \[ [\tuple{\pi,\alpha},\id]\colon [\tuple{T,f},\tuple{P,g}]\to  [\tuple{S,e},\tuple{P,g}] \Mdot\]
We define the underlying simplicial relation to be $\id[P]$, and define the action on outcomes by ``precomposition with $\tuple{\pi,\alpha}$''.
More specifically, given an outcome of $[T,P]$ at $x\in X_P$, of the form $\tuple{U,\Ev_T(U)\to O_{P,x}}$,
we send it to the composite $\tuple{\,\pi(U), \; \Ev_S(\pi(U))\xrightarrow[]{\alpha_U} \Ev_T(U)\to O_{P,x}\,}$. To see that this cooperates with the predicates these scenarios are equipped with, note that a possible outcome at $\sigma\in\Sigma_{[T,P]}$ corresponds to a function $\Ev_T(U)\to\Ev_P(\sigma)$ with $U\in\Sigma_T$ that sends outcome assignments allowed by $f$ to assignments allowed by $g$. Whenever this is the case, the composite $\Ev_S(\pi(U))\xrightarrow[]{\alpha_U} \Ev_T(U)\to O_{P,x}$ satisfies $\pi(U)\in\Sigma_S$ and it sends outcomes allowed by $e$ to outcomes allowed by $g$. Thus  $[\tuple{\pi,\alpha},\id]$ indeed is a morphism $[\tuple{T,f},\tuple{P,g}]\to  [\tuple{S,e},\tuple{P,g}]$. 

Now that $[-,-]$ has been shown to be functorial in each variable, it remains to show that it is a  bifunctor $(\ScengDet)\op\times \ScengDet\to \ScengDet$. But this is clear, since pre- and postcomposition commute with each other, \ie $(f\circ g)\circ h=f\circ (g\circ h)$ for all functions $f,g,h$.

We have now defined the bifunctor $[-,-]$. We next move to defining additional structure this bifunctor has, after which we check that this structure satisfies the required axioms. As our unit object, we choose the pair $\tuple{\Zero,\mathsf{t}}$, where $\Zero$ is the empty scenario and $\mathsf{t}$ is the predicate corresponding to the value $1\in\Ev_\two(*)$.
For any $\tuple{S,g}$, there is an isomorphism $i_{S,g}\colon \tuple{S,g}\to [\tuple{\Zero,\mathsf{t}},\tuple{S,g}]$ whose underlying simplicial relation is just the identity relation and which identifies an outcome in $\Ev_S(X)$ with the function  $\Ev_\Zero(\emptyset)=\enset{*}\to \Ev_S(x)$ (an outcome of $[\tuple{\Zero,\mathsf{t}},\tuple{S,g}]$ at $x$) having it as its value at the unique point of the domain. These isomorphisms are clearly natural in $\tuple{S,g}$.

We define a collection of morphisms $j_{S,g}\colon \tuple{\Zero,\mathsf{t}}\to [\tuple{S,g},\tuple{S,g}]$ by setting the underlying simplicial relation to be the empty relation. The function $\Ev_{\Zero}(\emptyset)\to\Ev_{[S,S]}(x)$ then sends the unique element of $\Ev_{\Zero}(\emptyset)$ to the function $\id\colon\Ev_S(x)\to\Ev_S(x)$. The morphisms $j_{S,g}$ should be dinatural in $\tuple{S,g}$, which amounts to checking the commutativity of
 \[\begin{tikzpicture} 
         \matrix (m) [matrix of math nodes,row sep=2em,column sep=5em,minimum width=2em]
         {
          \tuple{\Zero,\mathsf{t}} & \left[\tuple{S,g},\tuple{S,g}\right]  \\
          \left[\tuple{T,h},\tuple{T,h}\right] & \left[\tuple{S,g},\tuple{T,h}\right]  \\};
         \path[->]
         (m-1-1) edge node [left] {$j_{T,h}$} (m-2-1)
                edge node [above] {$j_{S,g}$} (m-1-2)
         (m-1-2) edge node [right] {$[\id,\tuple{\pi,\alpha}]$} (m-2-2)
         (m-2-1) edge node [below] {$[\tuple{\pi,\alpha},\id]$} (m-2-2);
  \end{tikzpicture}\]
for any $\tuple{\pi,\alpha}\colon\tuple{S,g}\to\tuple{T,h}$. Both morphisms have the empty relation as their underlying simplicial relation. Moreover, both of these morphisms have the same functions on outcomes: for $x\in X_T=X_{[S,T]}$ we get the function that sends the unique element of $\Ev_{\Zero}(\emptyset)$ to $\Ev_S(\pi(x))\xrightarrow[]{\alpha_x}\Ev_T(x)$.

Next we need to define morphisms $L_{S,g,T,h}^{P,f}\colon [\tuple{S,g},\tuple{T,h}]\to [[\tuple{P,f},\tuple{S,g}],[\tuple{P,f},\tuple{T,h}]]$ and check that they are natural in $\tuple{S,g}$, $\tuple{T,h}$ and dinatural in $\tuple{P,f}$. The intuition is that these morphisms internalise the composition operation on morphisms. We define $L_{S,g,T,h}^{P,f}$ as follows. As on both sides the underlying simplicial complex is (isomorphic to) $\Sigma_T$, we can set the underlying simplicial relation to be the identity. To define the functions on outcomes,
we need for each $x\in X_T$, a function that sends elements of $\Ev_{[S,T]}(x) = O_{[S,T],x}$ to elements of $O_{[[P,S],[P,T]],x}$.

Given an element of $O_{[S,T],x}$, \ie a tuple $\tuple{U,\fdec{k}{\Ev_S(U)}{O_{T,x}}}$, we send it to $\tuple{U,\tilde{k}} \in O_{[[P,S],[P,T]],x}$ where
$\fdec{\tilde{k}}{\Ev_{[P,S]}(U)}{O_{[P,T],x}}$
maps an element $\family{\tuple{V_y,\fdec{m_y}{\Ev_P(V_y)}{O_{S,y}}}}_{y\in U}$ of $\Ev_{[P,S]}(U)$
to the element of $O_{[P,T],x}$ given by
\[\tuple{\;\;V \defeq \bigcup_{y \in U} V_y\;,\;\;\; k \circ \tuple{m_y \circ \rho_y}_{y \in U}\;\;}\]
where the function is the composite at the bottom of the following diagram:
\[\begin{tikzcd}
& \Ev_P(V_y)  \arrow[r,above,"m_y"] & O_{S,y} \\
\Ev_P(V) \arrow[ru,"\rho_y"]  \arrow[rr,dashed,"\tuple{m_y\circ \rho_y}_{y \in U}"] &  & \Ev_S(U) = \prod\limits_{y \in U}O_{S,y} \arrow[u]
\arrow[rr,"k"] & & O_{T,x} \\
\end{tikzcd}
\Mdot\]
Thus composition with $k$ defines an element of $O_{[[P,S],[P,T]],x}$. As morphisms of scenarios that are compatible with the structure predicates compose, the morphism of scenarios $[S,T]\to [[P,S],[P,T]]$ is indeed a morphism $[\tuple{S,g},\tuple{T,h}]\to [[\tuple{P,f},\tuple{S,g}],[\tuple{P,f},\tuple{T,h}]]$.
Naturality of $L_{S,g,T,h}^{P,f}$ in $\tuple{S,g}$ and  $\tuple{T,h}$ follows from associativity of function composition. Dinaturality in $\tuple{P,f}$ amounts to the commutativity of 
 \[\begin{tikzpicture} 
         \matrix (m) [matrix of math nodes,row sep=2em,column sep=7em,minimum width=2em]
         {
          \left[\tuple{S,g},\tuple{T,h}\right] & \left[\left[\tuple{P,f},\tuple{S,g}\right],\left[\tuple{P,f},\tuple{T,h}\right]\right] \\&\\
         \left[\left[\tuple{Q,e},\tuple{S,g}\right],\left[\tuple{Q,e},\tuple{T,h}\right]\right]  & \left[\left[\tuple{P,f},\tuple{S,g}\right],\left[\tuple{Q,e},\tuple{T,h}\right]\right]  \\};
         \path[->]
         (m-1-1) edge node [left] {$L_{S,g,T,h}^{Q,e}$} (m-3-1)
                edge node [above] {$L_{S,g,T,h}^{P,f}$} (m-1-2)
         (m-1-2) edge node [right] {$[\id[], [\tuple{\pi,\alpha},\id]]$} (m-3-2)
         (m-3-1) edge node [below] {$[[\tuple{\pi,\alpha}],\id],\id]$} (m-3-2);
  \end{tikzpicture}\]
for any $\fdec{\tuple{\pi,\alpha}}{\tuple{Q,e}}{\tuple{P,f}}$, which again follows from associativity of function composition.

Next, we check that this data satisfies the five axioms required of a closed category. Here we refer to the numbering CC1--CC5 used in~\cite{manzyuk:closedcats}.

Axiom CC1 amounts to commutativity of 
 \[\begin{tikzpicture} 
         \matrix (m) [matrix of math nodes,row sep=2em,column sep=5em,minimum width=2em]
         {
          \tuple{\Zero,\mathsf{t}} & \left[\tuple{S,g},\tuple{S,g}\right]  \\
          &          \left[\left[\tuple{T,h},\tuple{S,g}\right],\left[\tuple{T,h},\tuple{S,g}\right]\right]  \\};
         \path[->]
         (m-1-1) edge node [below] {$j_{{[\tuple{T,h},\tuple{S,g}]}}$} (m-2-2)
                edge node [above] {$j_{S,g}$}     (m-1-2)
         (m-1-2) edge node [right]{$L_{S,g,S,g}^{T,h}$} (m-2-2);
  \end{tikzpicture}\]
 which boils down to the fact that composing with $\id$ keeps everything else fixed. 
 The diagram 
  \[\begin{tikzpicture} 
         \matrix (m) [matrix of math nodes,row sep=2em,column sep=5em,minimum width=2em]
         {
        \left[\tuple{S,g},\tuple{T,h}\right] & \left[\left[\tuple{S,g},\tuple{S,g}\right],\left[\tuple{S,g},\tuple{T,h}\right]\right]  \\
          &  \left[\tuple{\Zero,\mathsf{t}}, \left[\tuple{S,g},\tuple{T,h}\right]\right]        \\};
         \path[->]
         (m-1-1) edge node [below] {$i_{{[\tuple{S,g},\tuple{T,h}]}}$} (m-2-2)
                edge node [above] {$L_{S,g,T,h}^{S,g}$}     (m-1-2)
         (m-1-2) edge node [right]{$[j_{S,g},\id]$} (m-2-2);
  \end{tikzpicture}\]
  commutes essentially for the same reason, establishing CC2. 
 Axiom CC3 boils down to associativity of composition in $\ScengDet$, and CC4 is straightforward to check. 
 
The final and perhaps most important axiom, CC5, asserts that the function that sends
$\fdec{f}{\tuple{S,g}}{\tuple{T,h}}$ to the composite $\tuple{\Zero,\mathsf{t}}\xrightarrow{j_{S,g}}[\tuple{S,g},\tuple{S,g}]\xrightarrow{[\id,f]}[\tuple{S,g},\tuple{T,h}]$ defines a bijection between morphisms $\tuple{S,g}\to\tuple{T,h}$ and morphisms $\tuple{\Zero,\mathsf{t}}\to [\tuple{S,g},\tuple{T,h}]$. This essentially follows from Proposition~\ref{prop:modelson[S,T]}: deterministic morphisms $\tuple{\Zero,\mathsf{t}}\to [\tuple{S,g},\tuple{T,h}]$ that cooperate with the structure predicates
 are the same thing as morphisms $\tuple{S,g}\to\tuple{T,h}$.
\end{proof}

\begin{corollary}
 $[-,-]$ makes $\Sceng$ and $\Sceng_\BB$ into closed categories.
\end{corollary}

\begin{proof}
We sketch the proof for $\Sceng$, the case of $\Sceng_\BB$ being similar. We extend \[[-,-]\colon (\ScengDet)\op\times \ScengDet\to \ScengDet\] to a functor \[[-,-]\colon (\Sceng)\op\times \Sceng\to \Sceng\] by keeping the definition on objects fixed, and defining \[\left[\sum_i r_if_i,\sum_js_jg_j\right]\defeq\sum_{i,j}r_is_j\left[f_i,g_j\right]\Mdot\] We define the morphisms $j,i,L$ as before, and their naturality (and dinaturality) for the extended functor follows from them being natural (or dinatural) in the first place.%
    \footnote{Most of this should follow from the fact that changing the basis of enrichment is a 2-functor and in our case preserves the duality involutions, so one should automatically get a functor $[-,-]$ and the required natural transformations, leaving only dinaturality and some of the axioms to be checked by hand. However, we are not aware of general results guaranteeing that change-of-basis preserves closed structure, so we sketch the hands-on proof.}
Similarly, the axioms CC1--CC4 follow from those holding in $\ScengDet$. Finally, we wish to check axiom CC5, showing that postcomposition with $[\id,\sum_i r_if_i]=\sum_ir_i[id,f_i]$ defines a bijection between probabilistic morphisms $\tuple{S,g}\to\tuple{T,h}$ and probabilistic morphisms $\tuple{\Zero,\mathsf{t}}\to [\tuple{S,g},\tuple{T,h}]$. As we already had a bijection in the deterministic case, and since composition is convex, a convex mixture of morphisms $\tuple{S,g}\to\tuple{T,h}$ corresponds bijectively to a convex mixture of morphisms $\tuple{\Zero,\mathsf{t}}\to [\tuple{S,g},\tuple{T,h}]$, which completes the proof.
\end{proof}

\begin{remark}
 The category $\Sceng_\BB$ is isomorphic to $\cat{Emp}^{\leq}_\BB$. To see this, note that one can think of an object $\tuple{S,g}$ of $\Sceng_\BB$ as a pair $\tuple{S,e_g}$ where $e\in\EMP_\BB(S)$ induces $g$. Moreover, saying that a possibilistic procedure $f\colon S\to T$ is a morphism $\tuple{S,g}\to\tuple{T,h}$ is equivalent to saying that $\EMP_\BB(f)e_g$ is included in $e_h$, \ie that $f$ defines a weak simulation $e_g\to e_h$, \ie a morphism $e_g\to e_h$ in  $\cat{Emp}^{\leq}_\BB$.
\end{remark}

\section{What is to be done?}\label{sec:outlook}

With the results of the previous sections in hand, we now turn to a discussion
of some of the open questions and future research directions that naturally suggest
themselves.
We also offer some unpolished thoughts on how the sheaf and
resource approach to contextuality connects to other areas of research to
which Samson has made profound contributions.

\subsection{Using the closure}

Having a closed structure on a category lets one turn relativised questions, \ie those about morphisms, into questions about objects. This can be especially beneficial when the techniques for studying the objects are relatively well developed.
In our case, it suggest studying the resource theory of contextuality using the plethora of tools developed for studying contextuality, \eg cohomology~\cite{cohomology-of-contextuality,contextualitycohomologyparadox},
the contextual fraction~\cite{abramsky2017contextual}, or Bell inequalities. In particular, one is left wondering if, for suitable choices of scenarios, the logical Bell inequalities of~\cite{abramskyhardy:logical} have further logical structure that would help in answering questions of simulability.

\subsection{Variants of the main results}\label{ssec:variantsofmainresult}

In Section~\ref{sec:main-result} we presented a body of results about probabilistic procedures $\fdec{f}{S}{T}$ and their actions on (the convex sets of) probabilistic empirical models, $\fdec{\EMP(f)}{\EMP(S)}{\EMP(T)}$.
It is natural to ask whether these results carry over as one varies the categories of procedures and corresponding empirical models.
There are three main axes of variation: one could replace probabilities by possibilities, one could allow for adaptive procedures as studied in~\cite{ourlics:comonadicview}, and finally, one could work with scenarios equipped with structure predicates as in Section~\ref{sec:closure}.
Each results in interesting questions that aren't immediately answered by the above.

While some of the main results do carry over, the proofs of some sharper results do not.
Notably, our strategy for proving Theorem~\ref{thm:characterizingdeterministicprocedures}, which provides a characterisation of the actions on empirical models that are induced by deterministic procedures, runs into trouble in all three cases.

We will explain the specific issues arising along each axis of variation for our notion of procedure.
In summary, in the possibilistic case and in the variant with structure predicates, the issue is that a suitably structure-preserving map between (structured sets of) empirical models is not necessarily determined by its action on the deterministic models. In the adaptive case, the issue is with Lemma~\ref{lem:leastsubset}.

\subsubsection*{Possibilistic procedures}
In the possibilistic case, we wish to know when a map $F\colon\EMP_\BB(S)\to\EMP_\BB(T)$ is induced by a possibilistic procedure $S\to T$. We might as well start by assuming that $F$ preserves possibilistic sums as that is a necessary condition.
However, while this is an analogue of convexity, there is no obvious way to guarantee that $F$ is determined by its restriction to $\Ev(X_S)$. In other words, we do not have a possibilistic variant of Theorem~\ref{thm:convexfnsdeterminedatsections}.
This spells trouble for the nontrivial direction of Theorem~\ref{thm:characterizingdeterministicprocedures}.

On the other hand, possibilistic procedures $S\to T$ readily give rise to possibilistic models on $[S,T]$, so we are still able to prove a version of Theorem~\ref{thm:mainthm} in the possibilistic case.

\subsubsection*{Adaptive procedures}
The main obstacle when working with adaptive procedures is that the analogue of Lemma~\ref{lem:leastsubset} fails:
given a function $F\colon\Ev(X_S)\to Y$, there is no clear candidate for the optimal adaptive procedure to implement $F$. More abstractly, Lemma~\ref{lem:leastsubset} stems from the fact that each subset $U\subseteq X_S$ defines an equivalence relation $\sim_U$ on $\Ev(X_S)$ by setting $s\sim_U t$ whenever $s|_U=t|_U$, and such equivalence relations form a sublattice of all equivalence relations on $\Ev(X_S)$. Similarly, any measurement protocol on $X_S$ induces an equivalence relation on $\Ev(X_S)$ by setting $s\approx t$ when running the measurement protocol on $s$ and $t$ results in the same outcome. However, there are relatively small examples showing that such equivalence relations no longer form a sublattice of the lattice of all equivalence relations.
This makes it difficult to prove an adaptive variant of Theorem~\ref{thm:characterizingdeterministicprocedures} as our current proof relies heavily on Lemma~\ref{lem:leastsubset} in order to build the desired procedure.

Nonetheless, we are still able to establish adaptive variants of Theorem~\ref{thm:mainthm} and the results from Section~\ref{sec:closure} regarding closure.
Indeed, one can easily write down an adaptive version of $[S,T]$, and show that deterministic adaptive procedures correspond exactly to its global assignments that satisfy the predicate checking for simpliciality. Thus, probabilistic and adaptive global sections give rise to non-contextual models over (the adaptive variant of) $[S,T]$ that satisfy the simpliciality predicate, and moreover this adaptive version gives rise to closed structures on the adaptive versions of the categories in Definition~\ref{def:scenswithgames}.

\subsubsection*{Scenarios with structure predicates}

If given a function $\EMP(S)\to\EMP(T)$ that preserves convex combinations, deterministic models,
and also takes models that satisfy $g$ to models that satisfy $h$, we could readily apply Theorem~\ref{thm:characterizingdeterministicprocedures} to understand when it arises from a procedure $\tuple{S,g}\to\tuple{T,h}$.
However, if we are simply given a function $\fdec{F}{\EMP(S,g)}{\EMP(T,h)}$ that preserves convex combinations, the situation is less clear.
Indeed, it may happen that no deterministic model satisfies $g$ and then there is no clear sense in which $F$ is determined by a function $\Ev(X_S)\to\EMP(T)$.
In particular, if $g$ and $h$ are given by strongly contextual models $e_g$ and $e_h$ with minimal supports, then $\EMP(S,g)$ and $\EMP(T,h)$ are singletons and the unique function $\EMP(S,g)\to\EMP(T,h)$ preserves convex combinations vacuously.
Thus characterising which convex functions are induced by procedures $\tuple{S,g}\to\tuple{T,h}$ also gives a criterion for deciding whether $e_g$ simulates $e_h$, which in general is a difficult problem.

\subsection{Monoidal closure?}

Given the closed structure on $\Sceng$, a natural question is whether it is part of a monoidal closed structure. However, it seems like getting a monoidal structure would require:
(i) allowing adaptive morphisms, and (ii) having a ``directed'' tensor product in the sense that adaptive protocols on $S\otimes T$ should start on $S$ and consume it before moving on to $T$, never returning back.

To see what suggests this picture, consider first the natural candidate for an evaluation map $[S,T]\otimes S\to T$: to simulate something in $T$, one first measures the same measurement of $[S,T]$, the answer to which determines what to do in $S$ and how to interpret the outcome. This evaluation procedure is adaptive and directed. Moreover, any directed adaptive procedure $S\otimes T \to V$ seems to correspond to a procedure $S \to [T,V]$: if one just measures the $S$-half, the rest of the procedure tells what to do in regard to $T$, resulting in an outcome of $[T,V]$. Conversely, any morphism $S\to [T,V]$ seems to correspond to an adaptive and directed procedure $S\otimes T\to V$. To make this precise one might consider working with scenarios that are equipped with a preorder on measurements, and with adaptive protocols that only go forward along this preorder. Moving to such a setup is tempting not only because of potentially nice categorical properties, but also because it enables one to express other kinds of contextuality that do not quite fit in the Abramsky--Brandenburger framework as is. Indeed, one of the present authors has elsewhere considered moving to such a framework to discuss \stress{sequential} contextuality, as manifested in scenarios in which contexts are sequences of operations rather than sets of compatible one-shot measurements \cite{mansfield2018quantum,emeriau2020quantum}, and considered a cruder route to capturing violations Legget-Garg inequalities and macrorealism \cite{leggett1985quantum} by adding some extra structure by hand to the framework \cite{mansfield2017unified}. A similar approach is developed in~\cite{Gogiosopinzani:definitecausality}. However, incorporating morphisms of scenarios in that setting remains to be done.

\subsection{Towards duality via (generalised) partial Boolean algebras}

Recent work has shown how some of the concepts of the
sheaf-theoretic approach to contextuality can be formulated
in the language of partial Boolean algebras \cite{abramsky2021logic}.
This provides a dual algebraic-logical picture, in the sense
of Stone duality.

One limitation of this result, however, is that it only works for graphical measurement scenarios, \ie those whose compatibility structure is generated by a binary compatibility relation,
so that the simplicial complex of contexts arises as the clique complex of a graph. Extending this to all scenarios will require working with a generalisation of partial Boolean algebras, perhaps as discussed in~\cite{czelakowski1979} . Such a structure could be obtained by dualising the sheaf of events $\Ev$ to a copresheaf of Boolean algebras.
There is some hope that this could also be built form the set of procedures $S\to\dice{2}$, but this remains to be checked.

Another observation is that the partial Boolean algebras that correspond to measurement scenarios are in a sense freely built from the Boolean algebras corresponding to each measurement plus a compatibility relation. The theory of partial Boolean algebras is, however, much richer.
A speculative suggestion relates to the addition of the predicate $g$ to a scenario $S$. Its effect, at each context, is to pick out a subset of events. If one thinks of this as a clopen subset of a Stone space, then on the logico-algebraic side it corresponds
to a filter (or ideal). One may thus conjecture that the passage from $S$ to $\tuple{S,g}$ might correspond to the taking quotients of (generalised) partial Boolean algebras.

More generally, one might wonder whether there is a Stone-type duality at work. The scenario $\dice{2}$, which could be identified with the two-element discrete space, could conceivably act as the dualising object on the topological or model side. Note that, suggestively, procedures $S\to\dice{2}$ are naturally endowed with a logico-algebraic structure.

This also suggests a dual question to question \eqref{q:proc}. Instead of asking which (forward) state transforms arise from procedures,
one could inquire which (backward) predicate transforms do.

\subsection{What else can one do with predicates?}

Changing our basic objects to be pairs $\tuple{S,g}$ instead of just a scenario $S$ serves the purpose of restricting the kinds of behaviours allowed over an object. As such, it is reminiscent of the definition of basic objects in $\cat{SProc}$, a category Samson introduced as a model of concurrency~\cite{abramsky1996specification,abramsky1996interaction}.
Can this superficial similarity be pushed further? In particular, how much further structure does $\Sceng$ or its variants have? Is it \eg $*$-autonomous?

In~\cite{ourlics:comonadicview} we formalise adaptive procedures via a comonad $\MP$ that builds a new scenario $\MP(S)$, the measurements of which correspond to adaptive measurement protocols over $S$. Similarly, one could start by requiring that procedures have an underlying simplicial function rather than a relation, and then extend the category via a comonad $F$ that takes a scenario $S$ to the scenario $F(S)$ whose measurements are joint measurements in $S$.
Both of these comonads operate well with empirical models. If $e$ is a model on $S$ then there are induced empirical models $\MP(e):\MP(S)$ and $F(e):F(S)$. And moreover, these comonads induce comonads on the category of empirical models. However, there is a slight aesthetic issue in that not all empirical models on $\MP(S)$ (or on $F(S)$) are of this form. This is because the bare structure of a measurement scenario does not enforce on all models the intuition that a measurement of the new scenario corresponding to a derived measurement in the original one ought to behave as such. Our expectation is that this could be remedied by equipping $\MP(S)$ and $F(S)$ with structure predicates enforcing this condition. For the adaptive case, the structure predicate should itself be adaptive. 

\subsection{Possibilistic polytope}

Proposition~\ref{prop:possibilisticgamesasmodels} implies that the preorder of possibilistic morphisms $S\to \two$ is (dually?) equivalent to that of possibilistic models on $S$.
However, different operations seem natural on different sides of this equivalence. It is natural to take a union (\ie a possibilistic sum) of  predicates $\fdec{g}{S}{\two}$, which restricts the set of models that satisfy it. On the other hand, the union or Boolean sum operation for possibilistic models increases the set of models that satisfies the corresponding predicate. Thus one might hope that results on one side lead to new results on the other. Given that $\EMP_B(S)$ has been extensively studied in~\cite{abramsky2016possibilities}, one may hope either to leverage the results therein to better understand $\Scen_\BB(S,\two)$, or alternatively, to study  $\Scen_\BB(S,\two)$ in order to answer questions left open therein. In particular, can one understand the image of possibilistic collapse $\EMP(S)\to\EMP_{\BB}(S)$ by studying $\Scen_\BB(S,\two)$?

\subsection{Other questions}

The fact that procedures $S\to T$ give rise to non-contextual models on $[S,T]$ suggests thinking of contextual models  $e:[S,T]$ as ``contextual simulations'' $S\to T$. The aforementioned evaluation map gives an intuition about how to think of these operationally: given a measurement in $T$ to simulate on $S$, query the model $e$ on the same measurement -- the outcome of this will then indicate what to measure in $S$ and how to interpret it. This seems very closely related to another possible formulation of contextual simulations between empirical models, suggested at the end of~\cite{ourlics:comonadicview}: namely, contextual morphisms $d\to e$ could be formalised as adaptive maps $d\otimes c\to e$ without requiring that $c$ be non-contextual. Key differences are that this latter notion is formulated at the level of models rather than at the level of scenarios, and that the adaptivity need not be directed. Despite this, these could still turn out to be similar in terms of expressive power, and may give rise to interesting notions when one still restricts the amount of contextuality allowed, by \eg requiring that the contextual resources be quantum-realisable. 

Another natural question for the computer scientist would be to go beyond single-use black boxes.
Here we mean not just allowing for adaptivity but for true multiple use, with an internal state (as in Mealy machines). In the physics side, this would correspond to non-destructive measurement, or measurement with a next state. It might be interesting to extend the formalism to deal with such scenarios.

\section{Epilogue: A game of cat-and-mouse, a personal thank-you}

A recurrent tip of Samson's is to read Gian-Carlo Rota's \textit{Ten lessons I
wish I had been taught} \cite{rota1997ten}. There, amid an assortment of
sage advice, can be found a passage about becoming an institution and the
dubious expectations that are bestowed upon one in that event. One could almost look at
Samson's career as a game of cat-and-mouse against this prophecy. Rather than
settle in to a comfortable and earned existence as `a piece of period
furniture', time after time his relentless scientific curiosity has impelled
him to go forth to sow and explore new fields -- only to quickly find himself
established, there too, as something of an institution.


But plenty has been said of Samson's outstanding scientific contributions, to
logic and otherwise, and by others much better qualified than ourselves. So let
us take this opportunity to end on a more personal note.

Would-be institutions seem to come in different flavours: there are those that,
overbearing, tend to crush and stifle, and there are those that tend to be `rotten and
rotting others', and yet there are also those who help the people around them to flourish
and thrive to the best of their abilities. Samson is certainly one of the
latter. He has been a central formative influence in our development as
(computer?) scientists, and we count ourselves privileged for having worked,
and for continuing to work, closely with him. Even as students and fledgling
researchers we felt valued and taken seriously, and were afforded the space and
encouragement to develop our ideas, with the occasional crucial nudge in the
right direction. It may seem little, but it ain't.

We have witnessed up close, on enjoyable afternoons spent exchanging in
Samson's office, or over the occasional cheeky pint, his impressive breadth of
knowledge, his ability to see further and to spot or establish connections
between disparate fields, even across traditional disciplinary boundaries. We
have been led to recognise those boundaries not as ontological but as imagined
and thus re-imaginable. And in his company we have experienced an environment
in which common language can be found and shared among people of very different
backgrounds. Testament to this capacity for bringing people together is the present
set of authors, including by original academic background a computer scientist,
a mathematician, and a physicist.

But alongside his bird-like vision, there is also a frog-like attention to
detail \cite{dyson2009birds}. There are
the sudden bursts of activity, when e-mails budding with newly formed ideas
are fired afresh every few minutes, when he's engrossed in some intricate detail of
work -- or swimming with purpose in the milk churn, to borrow
a metaphor from one of Samson's favourite quotes found in the preface to Littlewood's miscellany \cite{littlewood1986littlewood}. And there are the times
when he pins you down on a fine point, forcing you to be precise and organise
your own thoughts. Indeed, it is through these sorts of exchanges that the milk
is often churned to butter.\footnotemark

\footnotetext{Littlewood wrote of his collaboration with Mary
Cartwright:
\textit{Two rats fell into a can of milk. After swimming for a time one of
them realised his hopeless fate and drowned. The other persisted, and
at last the milk was turned to butter and he could get out.}}

While as researchers we may each treasure those moments when, as in
Wordsworth's immortal lines about Newton, we find ourselves `voyaging through strange seas
of thought, alone', with Samson one also learns to see science as a
conversation. In that spirit we also wish to acknowledge the `sheaf team' and
the many friends and collaborators who have been a part of the conversations
and ideas that have nourished the research discussed in this chapter.

We would like to thank Samson and the editors, Mehrnoosh and Alessandra,
for giving us the opportunity to contribute a chapter to this volume,
and for their patience with us during a
tumultuous and drawn-out preparation period.
Doing so has been for us a source of great pleasure.
We hope, Samson, that you will find it equally enjoyable to read.

And in the end, so the saying goes, curiosity killed the cat\dots


\section*{Acknowledgements}
This work was in part carried out while RSB was based at the School of Informatics,
University of Edinburgh and while SM was based at the Paris Centre for Quantum Computing,
Institut de Recherche en Informatique Fondamentale, University of Paris with financial
support from the Bpifrance project RISQ.
RSB acknowledges financial support from EPSRC -- Engineering and Physical Sciences Research Council, EP/R044759/1, \textit{Combining Viewpoints in Quantum Theory (Ext.)}, and from FCT -- Fundação para a Ciência e a Tecnologia, CEECINST/00062/2018.

\bibliographystyle{spmpsci}
\bibliography{chapter_samsky}
\end{document}